\RequirePackage{fix-cm}

\documentclass[smallextended]{svjour3}       % onecolumn (second format)
\smartqed  % flush right qed marks, e.g. at end of proof

\usepackage{graphicx}
\usepackage{amsmath}
\usepackage{amssymb}
\usepackage{mathrsfs}
\usepackage{mathtools}
\usepackage{comment}
\usepackage{graphicx}
\usepackage[caption=false]{subfig}
\usepackage{footnote}
\usepackage{textcomp}
\usepackage{epstopdf}
\usepackage[maxfloats=99]{morefloats}
\usepackage{cases}
\usepackage{color}
\usepackage{cite}
\usepackage{bm}
\usepackage{enumitem}

\newcommand{\Rb}{\mathbb{R}}

\newcommand{\Zb}{\mathbb{Z}}
\newcommand{\Nb}{\mathbb{N}}

\newcommand{\Lc}{\mathcal{L}}

\newcommand{\Vc}{\mathcal{V}}

\newcommand{\ri}{{\mathrm{i}}}

\newcommand{\im}{\mathrm{Im}}

\newcommand{\eps}{\varepsilon}

\newcommand{\inprod}[3]{\left\langle #1, #2 \right\rangle_{#3}  }

\newcommand{\we}{W^{\varepsilon}}
\newcommand{\invpi}{\frac{1}{(2\pi)^d}}

\newcommand{\epsh}{\frac{\varepsilon}{2}}
\newcommand{\half}{\frac{1}{2}}

\newcommand{\LdC}{L^2(\Rb^d)}
\newcommand{\LddC}{L^2(\Rb^{2d})}

\newcommand{\Vp}{\widetilde{V}}
\newcommand{\Bp}{\widetilde{\phi}_T^\eps}
\newcommand{\phip}{\widetilde{\phi}^\eps}
\newcommand{\fp}{\widetilde{f}^\eps}

\newcommand{\rmd}{\mathrm{d}}
\newcommand{\rmb}{\mathrm{b}}
\newcommand{\rmI}{\mathrm{I}}
\newcommand{\rmS}{\mathrm{S}}
\newcommand{\rmW}{\mathrm{W}}
\newcommand{\rmL}{\mathrm{L}}
\newcommand{\rmR}{\mathrm{R}}

\begin{document}

\title{Semi-classical limit of an inverse problem for the Schr\"odinger equation \thanks{Support for this research was provided by NSF under DMS-1750488, and the University of Wisconsin-Madison, Office of the Vice Chancellor for Research and Graduate Education with funding from the Wisconsin Alumni Research Foundation.}
}
% \subtitle{Do you have a subtitle?\\ If so, write it here}


\author{Shi Chen         \and
        Qin Li %etc.
}


\institute{Shi Chen \at
              Department of Mathematics, University of Wisconsin-Madison, Madison, WI, 53706, USA \\
              \email{schen636@wisc.edu}           %  \\
           \and
           Qin Li \at
              Department of Mathematics and Wisconsin Institute for Discovery, University of Wisconsin-Madison, Madison, WI, 53706, USA \\
              \email{qinli@math.wisc.edu}
}

\date{Received: date / Accepted: date}

\maketitle

\begin{abstract}
  It is a classical derivation that the Wigner equation, derived from the Schr\"odinger equation that contains the quantum information, converges to the Liouville equation when the rescaled Planck constant $\eps\to0$. Since the latter presents the Newton's second law, the process is typically termed the (semi-)classical limit. In this paper, we study the classical limit of an inverse problem for the Schr\"odinger equation. More specifically, we show that using the initial condition and final state of the Schr\"odinger equation to reconstruct the potential term, in the classical regime with $\eps\to0$, becomes using the initial and final state to reconstruct the potential term in the Liouville equation. This formally bridges an inverse problem in quantum mechanics with an inverse problem in classical mechanics.
\end{abstract}

\section{Introduction}
The classical limit, or the semi-classical limit of quantum mechanics is the ability of quantum theory to recover, or partially recover classical mechanics when the rescaled Planck constant $\eps$ is considered negligible. More specifically, by setting $\eps\approx0$ in the Schr\"odinger equation, one is expected to recover the Newtonian's law of motion (Newton's second law) in the asymptotic limit.

The concept of linking quantum mechanics and classical mechanics was already in formulation in 1920s, and was presented by N. Bohr in his Nobel lecture under the name of ``correspondence principle''. Since then, there have been abundant studies on deriving and proving the classical limits. While the formal derivation using WKB expansion is relatively easy to show, the discontinuity in the limiting equation (Hamiltonian-Jacobi equation) makes the rigorous mathematics analysis hard to obtain. In~\cite{GeMaMaPo:1997,RyPaKe:1996,BaKoRy:2010}, the authors, by introducing Wigner measures, flipped the studies to the phase space and expanded out the singularity, upon which, the derivation of classical limit was made rigorous.

We investigate the problem in an inverse setup. Suppose a quantum system is modeled by the Schr\"odinger equation, and one can measure the initial and final state, can one reconstruct the potential term (the field) in the equation? Moreover, if the quantum system is in the classical regime, with $\eps\approx 0$, can we view this inverse problem as the inverse problem for the Newtonian motion? What is the connection between the inverse Schr\"odinger and the inverse Newton's law? These questions essentially come down to deriving the classical limit of the inverse problem for the Schr\"odinger equation.

It is a relatively big topic, and in this paper in particular, we confine ourselves to the linearized setting. Namely, we assume the potential term is close to a preset background potential, and we are interested only in reconstructing the perturbation term. Under this setting, both the inverse Schr\"odinger problem and the inverse Newtonian motion problem can be formulated as Fredholm integrals, and it is the representatives (or the kernels) of the integrals that reveal the perturbed potential information. The question of deriving the classical limit, when confined in linearized setting, becomes: are the two representatives asymptotically equivalent when $\eps\to0$ in some sense?

The problem is of great interest, not only for our mathematical curiosity, but also for its practical use.

Since the fundamental question of bridging quantum mechanics and classical mechanics is mathematically clear, it is very natural to seek for its correspondence in the inverse setting. Indeed, in what sense can one view the inverse Schr\"odinger problem and the inverse Newtonian motion problem equivalently? Or, is it possible for one problem to be more stable than the other? This type of stability increasing/decreasing problem recently attracts a large amount of attention for various sets of problems~\cite{ChLiWa:2018,LaLiUh:2019,Wa:1999,NaUhWa:2013}.

Practically, the Schr\"odinger equation is not only regarded as the fundamental model for quantum mechanics, but also emerges as the limit of the Helmholtz equation when dynamics in different dimensions is described at separate scales~\cite{Go:2005}, and thus serves as a fundamental model for the wave propagation (for a fixed high frequency) as well. There are abundant applications, in which high-frequency waves are sent to detect the media~\cite{Ba:2013,BaRe:2011,BaReUhZh:2011,Be:2011,Sz:2004}. Mathematically, this is to seek for reconstructing the speed of sound in the Helmholtz equation, which is to reconstruct the potential term in the Schr\"odinger equation. Moreover, the inverse Schr\"odinger problem is also a transformed version of the celebrated Calder\'on problem, arises from Electrical Impedance Tomography (EIT)~\cite{Ca:2006}. For these reasons, inverse Schr\"odinger problem has long been regarded as one of the most important inverse problems. Most of the studies, however, set the Planck constant in the Schr\"odinger equation to be an $O(1)$ value. This is not practical in many applications mentioned above. In the high frequency regime for the Helmholtz equation, or in the classical regime with the rescaled $\eps\to0$, the stability of the inverse problem may change, and it would be of great practical interests to predict the stability in these regimes, and to quantify the reconstruction error in terms of the rescaled $\eps$. Linking it to the inverse Newtonian motion is a natural strategy.

Despite the great importance of the problem, the theoretical study has been thin, even though it is mentioned a couple of times in the literature~\cite{KeKaSh:1956,Jo:2014,Jo:2013,No:1999}. Most of the studies formulate the problem as the (quantum) scattering problem. See also the geometric version for reconstructing the refraction index~\cite{Mu:1981,MoStUh:2015,HoMoSt:2018,Mo:2014}. The obstacles come from (a) the disparity of the technicalities used in deriving the classical limit, and in analyzing inverse problems, and (b) the disparity in analyzing the two different inverse problems (inverse Schr\"odinger and inverse Newtonian motion). In this paper, we take an initial attempt to bridge the two under the linearized setting, hoping to unveil some connections that could potentially serve as stepping stones for further investigation. We should mention, that when the media encodes randomness, the classical limit of the Schr\"odinger equation (or similarly the wave equation) is the linear Boltzmann equation (or the radiative transfer equation) that characterizes the dynamics of photons on the mesoscopic level. The associated inverse problem is highly related to imaging, and has been studied in different contexts~\cite{BaPi:2007,BaRe:2008,BaPiRy:2015,CaSc:2015,HoKrSc:2018}.

The paper is organized as follows. In Section~\ref{sec:semi_classical}, we review the derivation of classical limit for the Schr\"odinger equation. This is done through applying the Wigner transform. In Section~\ref{sec:inverse}, we utilize the linearization approach to set up the frameworks for Schr\"odinger, Wigner and Liouville inverse problems. The relations between and the three inverse problems are considered in Section~\ref{sec:connecting}, including the equivalence of the Schr\"odinger and the Wigner inverse problem, and the convergence from the Wigner to the Liouville inverse problem as $\eps\to0$. Numerical tests are exploited in Section~\ref{sec:num_test} to demonstrate the convergence from the Wigner to the Liouville inverse problem.

\section{The Classical limit of the Schr\"odinger equation}\label{sec:semi_classical}
\subsection{Schr\"odinger equation}
In this section we present some preliminary results that show the classical limit of the Schr\"odinger equation in the $\eps\to0$ regime.

For a nonrelativisitic single particle, the time-dependent Schr\"odinger equation in position basis writes as:
\begin{equation}\label{eqn:schr}
\begin{aligned}
\ri\eps\partial_t\phi^\eps = -\half\eps^2\Delta_x\phi^\eps + V(x)\phi^\eps\,,\quad x\in\Rb^d\,, \quad t>0\,, \\
\phi^\eps(0,x) = \phi_\rmI^\eps(x)\,,\quad x\in\Rb^d\,.
\end{aligned}
\end{equation}
This is derived assuming the Hamiltonian is $H = \frac{1}{2}|k|^2+V(x)$, a summation of kinetic and potential energies of the particles constituting the system. In the equation, $\phi^\eps$ is the wave function, $\eps>0$ is the rescaled Planck constant, and $V(x)$ is the potential term.

Some physical quantities can be calculated using $\phi^\eps$. For example, the particle density $\rho^\eps$ and current density $J^\eps$ are calculated by
\begin{equation*}
\rho^\eps(t,x) = \left|\phi^\eps(t,x)\right|^2\,,\quad J^\eps(t,x) = \eps \im \left( \overline{\phi^\eps(t,x)} \nabla_x \phi^\eps(t,x) \right)\,.
\end{equation*}
These present the probability and the probability flux of the particle found in some spatial configuration at some instant of time, according to the Copenhagen interpretation. Both quantities are quadratic functionals of $\phi^\eps(t)$, and it is straightforward to derive, from~\eqref{eqn:schr}, the following conservation law:
\begin{equation*}
\partial_t \rho^\eps + \nabla_x\cdot J^\eps = 0\,.
\end{equation*}

A more general definition of physical observables can be given using phase space symbols and Weyl quantization~\cite{Ho:85}. To make it more explicit, let $a(x,k)$ be a symbol, then using Weyl quantization, we can define a pseudo-differential operator $a^{W}(x,\eps D_x)$ whose action on $f(x)$ leads to:
\begin{equation}\label{eqn:Weyl_op}
(a^\rmW(x,\eps D_x)f)(x) = \invpi \int_{\Rb^{2d}}a\left( \frac{x+y}{2},\eps k \right) f(y) e^{\ri (x-y)k}\rmd y \rmd k\,,
\end{equation}
where $\eps D_x = -\ri \eps \nabla_x$. We then define the expectation value of the symbol $a$ to be a quadratic functional of wave function $\phi^\eps(t)$:
\begin{equation*}
a[\phi^\eps(t)] = \inprod{\phi^\eps(t)}{a^\rmW(x,\eps D_x)\phi^\eps(t)}{\LdC}\,,
\end{equation*}
where $\inprod{\cdot}{\cdot}{\LdC}$ denotes the inner product on $\LdC$.

The well-posedness theory of Schr\"odinger equation~\eqref{eqn:schr} is classical. For $V=V(x)$ being continuous and bounded, i.e., $V\in C_b(\Rb^d)$, the Hamiltonian operator $\hat{H}^\eps$ is
\begin{equation}\label{eqn:pseudo_def}
\hat{H}^\eps\phi^\eps = -\frac{\eps^2}{2} \Delta_x \phi^\eps(x) + V(x)\phi^\eps(x)\,.
\end{equation}
It maps functions in $H^2(\Rb^d)\subset \LdC$ to $\LdC$, and is self-adjoint. By Stone's theorem, the operator $\frac{1}{\ri\varepsilon} \hat{H}^\eps$ generates a unitary, strongly continuous semi-group on $\LdC$, which guarantees a unique solution to the Schr\"odinger equation~\eqref{eqn:schr}. Moreover, the $\LdC$ inner product is conserved in time:
\begin{equation}\label{eqn:inprod_conserve}
\inprod{\phi_1^\eps(t)}{\phi_2^\eps(t)}{\LdC} = \inprod{\phi_1^\eps(0)}{\phi_2^\eps(0)}{\LdC}\,,\quad \forall t>0\,,
\end{equation}
for $\phi_i^\eps(t), i = 1, 2$ both solve the Schr\"odinger equation~\eqref{eqn:schr}.

\subsection{Wigner transform and the classical limit}
Wigner transform is one of many approaches used to derive (semi-)classical limit of Schr\"odinger equations. The technique was explored in depth in~\cite{GeMaMaPo:1997}. Let $\phi_1^\eps(t)$ and $\phi_2^\eps(t)$ solve the Schr\"odinger equation, and we define the corresponding Wigner transform:
\begin{equation}\label{eqn:wigner_transform}
\we[\phi_1^\eps,\phi_2^\eps](t, x,k) = \invpi \int_{\Rb^d}e^{\ri ky} \phi_1^\eps \left(t, x - \epsh y\right) \overline{\phi_2^\eps}\left(t, x + \epsh y \right)\rmd y\,.
\end{equation}
Here $\overline{\phi^\eps}$ is the complex conjugate of $\phi^\eps$. This definition is essentially the Fourier transform of the density matrix
\begin{equation*}
\left\langle x-\frac{\varepsilon}{2}y\middle|\phi_1^\eps\right\rangle \left\langle \phi_2^\eps\middle|x+\frac{\varepsilon}{2}y \right\rangle
\end{equation*}
in the $y$ variable.

We furthermore abbreviate $\we[\phi^\eps,\phi^\eps]$ to be $\we[\phi^\eps]$. It is then straightforward to show that $\we[\phi^\eps]$ is real-valued.

Note that the Wigner transform loses the phase information: Changing $\phi^\eps(t)$ to $\phi^\eps(t) e^{\ri S(t)}$, one obtains the same corresponding Wigner function. Moreover, it is not guaranteed that $\we[\phi^\eps]$ is positive, and thus it does not serve directly as the particle density on the phase space. However, the quantum expectation of physical observables can be easily recovered using the Wigner function. Using the symbol defined in~\eqref{eqn:Weyl_op}, it can be shown~\cite{Ho:85} that
\begin{equation*}
a[\phi^\eps(t)]=\inprod{\phi^\eps(t)}{a^\rmW(x,\eps D_x)\phi^\eps(t)}{\LdC} = \int_{\Rb^{2d}}a(x,k)\we[\phi^\eps(t)] \rmd x \rmd k\,.
\end{equation*}

In particular, the first and second moments in $k$ of $\we[\phi^\eps]$ exactly recover the particle density $\rho^\eps(t)$ and the current density $J^\eps(t)$:
\begin{equation*}
\rho^\eps(t,x) = \int_{\Rb^d} \we[\phi^\eps(t)] (t,x,k) \rmd k\,,\quad J^\eps(t,x) = \int_{\Rb^d} k \we[\phi^\eps(t)] (t,x,k) \rmd k\,.
\end{equation*}

We now derive the equation for $\we[\phi_1^\eps,\phi_2^\eps]$, as summarized in the following lemma.
\begin{lemma}\label{lem:wigner}
Let $\phi_1^\eps(t)$ and $\phi_2^\eps(t)$ solve the Schr\"odinger equation~\eqref{eqn:schr}, and define
\begin{equation*}
f^\eps(t,x,k) = \we[\phi_1^\eps,\phi_2^\eps](t,x,k)\,.
\end{equation*}
Then $f^\eps$ satisfies the following Wigner equation:
\begin{equation}\label{eqn:wigner}
\begin{aligned}
\partial_t f^\eps + k\cdot\nabla_x f^\eps = \Lc_{V}^\eps[f^\eps]\,,\quad (x,k)\in\Rb^{2d}\,, \quad t>0\,, \\
f^\eps(0,x,k) = f_\rmI^\eps(x,k)\,,
\end{aligned}
\end{equation}
with $f_\rmI^\eps(x,k)$ being the Wigner transform of initial conditions $\phi_1^\eps(0)$ and $\phi_2^\eps(0)$, and the operator $\Lc_V^\eps$ is defined as:
\begin{equation}\label{eqn:wigner_Lc}
\Lc_V^\eps[f^\eps] = \ri \invpi \int_{\Rb^{2d}} \delta^\eps [V](x,y) f^\eps(x,p) e^{\ri y(k-p)}   \rmd y \rmd p\,.
\end{equation}
Here $\delta^\eps [V](x,y) = \frac{1}{\eps}\left[ V\left(x+\half\eps y\right) - V\left(x-\half\eps y\right) \right]$. Equivalently, one can also write
\begin{equation}\label{eqn:col_Deps}
\Lc_{V}^\eps[f^\eps] = \ri\int_{\Rb^{2d}} e^{\ri p(x-y)} V(y)D^\eps f^\eps(x,k,p)\rmd p \rmd y\,,
\end{equation}
where the term $D^\eps f^\eps$ is defined by
\begin{equation}\label{eqn:D_eps}
D^\eps f^\eps(x,k,p) = \frac{1}{\eps}\left[f^\eps\left(x,k+\half\eps p\right)-f^\eps\left(x,k-\half\eps p\right)\right]\,.
\end{equation}
\end{lemma}

We note that $\Lc^\eps_V$ is an operator that is anti-self-adjoint for all real-valued potential $V$. To see that, we first define
\begin{equation}\label{eqn:col_ker}
\Vc^\eps(x,k) = \ri \invpi \int_{\Rb^d} \delta^\eps [V](x,y) e^{\ri yk} \rmd y\,.
\end{equation}
This allows us to simplify~\eqref{eqn:wigner_Lc} to a convolution form
\begin{equation*}
\Lc_V^\eps[f^\eps] = \int_{\Rb^{2d}}  \Vc^\eps(x,k-p) f(x,p) \rmd p = \Vc^\eps \ast_k f^\eps \,.
\end{equation*}
Since $\Vc^\eps(x,-k) = -\overline{\Vc^\eps}(x,k)$, it is straightforward to see
\begin{equation*}
\inprod{\Vc^\eps \ast_k f_1^\eps}{f_2^\eps}{\LddC}
= - \inprod{f_1^\eps}{\Vc^\eps \ast_k  f_2^\eps}{\LddC}\,,
\end{equation*}
meaning:
\begin{equation}\label{eqn:self_ajoint_L}
\inprod{\Lc_V^\eps[f_1^\eps]}{f_2^\eps}{\LddC}= - \inprod{f_1^\eps}{\Lc_V^\eps[f_2^\eps]}{\LddC}\,.
\end{equation}

To derive the Wigner equation~\eqref{eqn:wigner}, one only needs to plug in the Schr\"odinger equation for both $\phi_1^\eps$ and $\phi_2^\eps$. The statement of the lemma is formal, but one can make it rigorous in $\LdC$. We omit the derivation from this paper, but refer interested readers to~\cite{GeMaMaPo:1997}.

The nice format of the Wigner equation makes it easy to obtain the classical limit. Indeed, formally, as $\eps\to0$, $\delta^\eps[V]\rightarrow y\cdot\nabla_x V$. Then according to the definition of the operator~\eqref{eqn:wigner_Lc}, we have
\begin{equation*}
\Lc_{V}^\eps[f^\eps] \to \nabla_x V \cdot \nabla_k f^\eps + O(\varepsilon^2)\,.
\end{equation*}
This means the asymptotic limit of~\eqref{eqn:wigner}, up to the truncation of $O(\eps^2)$, is the Liouville equation:
\begin{equation}\label{eqn:limiteqn}
\partial_t f + k\cdot\nabla_x f - \nabla_x V \cdot \nabla_k f = 0\,.
\end{equation}
Following the characteristic of this equation we have:
\begin{equation}\label{eqn:hamilton}
\dot{x} = k\,,\quad \dot{k} = -\nabla_x V(x)\,.
\end{equation}
This is exactly the same as the Newtonian law of motion generated by the Hamiltonian $H(x,k) = \frac{1}{2}|k|^2 + V(x)$.

This formal analysis can be made rigorous. Indeed in~\cite{GeMaMaPo:1997} the authors studied a general Hamiltonian system and derived the asymptotic limit for the Wigner equation. In our special case, it becomes:
\begin{theorem}\label{thm:semiclassical_limit}
Suppose the potential $V(x)$ satisfies
\begin{equation}\label{eqn:bounded_grad}
V(x) \in C^\infty(\Rb^d;\Rb): \quad |\partial_x^{\alpha} V(x)| \leq C_{\alpha} \quad \forall \alpha\in \Nb^d\,,
\end{equation}
then the Wigner transform $f^\eps(t,x,k)$ of $\phi^\eps(t)$, the solution to Schr\"odinger equation~\eqref{eqn:schr}, converges,in the weak-$\ast$ sense, locally uniformly in $t$ to the measure $f(t,x,k)$ that solves:
\begin{equation}\label{eqn:liouville_thm}
\partial_t f + k \cdot\nabla_x f - \nabla_x V \cdot \nabla_k f = 0\,, \quad f(0,x,k) = f_I(x,k)\,.
\end{equation}
The initial data $f_I$ is the weak-$*$ limit of Wigner transform of $\phi^\eps_\rmI(t)$.
\end{theorem}

\section{Three inverse problems}\label{sec:inverse}
We are now facing three equations: the original Schr\"odinger equation, the Wigner equation, and the Liouville equation as the classical limit of the Wigner equation. With respect to these three equations, we can formulate three inverse problems, all of which will be derived in this section.

We employ the same setup for the three inverse problems: we assume the equations are Cauchy problems without boundary constraints, and we confine ourselves to the linearized setting. This is to assume the potential term $V$ is close enough to a background $V_\rmb$. The given input is the initial data and one can measure the final state at a given time $T$. The to-be-reconstructed parameter is the potential term $V$ (or equivalently $\widetilde{V} = V-V_\rmb$).

We present the three inverse problems in the following three subsection respectively.

\subsection{A linearized inverse problem for the Schr\"odinger equation}\label{sec:lin_inv_Schr}
Recall the Schr\"odinger equation in $\Rb^d$ is
\begin{equation}\label{eqn:schr_full}
\ri\eps\partial_t\phi^\eps = -\half\eps^2\Delta_x\phi^\eps + V(x)\phi^\eps\,.
\end{equation}
Let the initial data be $\phi^\eps(0,x) = \phi_\rmI^\eps(x)$, and final data at $t=T$ be $\phi^\eps(T,x) = \phi_T^\eps(x)$. While the forward problem is to compute $\phi^\eps_T$ for every given $\phi_\rmI^\eps$, the inverse problem is to use $(\phi^\eps_\rmI,\phi^\eps_T)$ data pairs to reconstruct $V$. In other words, denoting
\begin{equation*}
\mathcal{M}^\eps_\rmS[V]\,:\;\phi^\eps_\rmI\to\phi^\eps_T\,,
\end{equation*}
the inverse problem is to use the map $\mathcal{M}^\eps_\rmS[V]$ to reconstruct $V$.

\begin{remark}
The reconstruction is at most unique up to a gauge transform. Indeed, let $\hat{H}^\eps$ be the Hamiltonian operator~\eqref{eqn:pseudo_def}, and define $H_n^\eps$ to be a new Hamiltonian operator
\begin{equation*}
\hat{H}_n^\eps = \hat{H}^\eps + \frac{2\pi\eps}{T}n \,, \quad n\in\Zb\,.
\end{equation*}
We further define the unitary semi-group generated by $\hat{H}^\eps$ and $\hat{H}_n^\eps$:
\begin{equation*}
U^\eps(t) = e^{-\ri t \hat{H}^\eps/\eps}\,,\quad U_n^\eps(t) = e^{-\ri t \hat{H}_n^\eps/\eps}\,,\quad t>0\,.
\end{equation*}
Clearly the two Hamiltonian operators are different, but $U^\eps(T) = U_n^\eps(T)$, for all $n\in\Zb$. This suggests that the initial-to-final map
\begin{equation*}
\mathcal{M}^\eps_\rmS[V] = \mathcal{M}^\eps_\rmS\left[V+\frac{2\pi\eps}{T}n\right]\,,
\end{equation*}
and thus the reconstruction cannot be unique.
\end{remark}

To derive the linearized version of the inverse problem, we assume there is a known background potential term $V_\rmb(x)$ such that
\begin{equation*}
\Vp(x) = V(x)-V_\rmb(x)
\end{equation*}
is much smaller than $V_b(x)$ in amplitude. We further write the background problem with the same initial condition:
\begin{equation}\label{eqn:schr_background}
\begin{aligned}
\ri\eps\partial_t\phi_\rmb^\eps = -\half\eps^2\Delta_x\phi_\rmb^\eps + V_\rmb(x)\phi_\rmb^\eps\,,\\
\phi_\rmb^\eps(0,x) = \phi^\eps_\rmI(x)\,.
\end{aligned}
\end{equation}
For a preset $V_\rmb$ and $\phi^\eps_\rmI(x)$, one can compute the equation for $\phi^\eps_\rmb(T,x) = \phi^\eps_{\rmb,T}(x)$.

Let $\phip = \phi^\eps-\phi_b^\eps$ be the perturbation of wave $\phi^\eps$, then by subtracting the equation~\eqref{eqn:schr_full} from~\eqref{eqn:schr_background} and omitting the higher order term $\Vp\phip$, we get the equation for the perturbation $\phip$
\begin{equation}\label{eqn:schr_perturbed}
\begin{aligned}
\ri\eps\partial_t\phip = -\half\eps^2\Delta_x\phip + V_\rmb(x)\phip + \Vp(x)\phi_\rmb^\eps\,,\\
\phip(0,x) = 0\,.
\end{aligned}
\end{equation}
Note that $\phip$ has trivial initial data and implicitly depends on the initial condition $\phi_\rmI^\eps$ through the background wave $\phi_\rmb^\eps$. Knowing the measured data $\phi^\eps_T(x)$, and the computed data $\phi^\eps_{\rmb,T}(x)$, we merely take the difference and define
\begin{equation}\label{eqn:schr_measured}
\Bp=\phip(T,x) = \phi^\eps_T(x)-\phi^\eps_{\rmb,T}(x)\,.
\end{equation}
The inverse problem now translates to reconstructing $\Vp$ using $(\phi_\rmI^\eps\,,\Bp)$ data pairs. To do so, we formulate the adjoint equation $\psi^\eps$ that solves:
\begin{equation}\label{eqn:schr_conjugate}
\begin{aligned}
\ri\eps\partial_t\psi^\eps = -\half\eps^2\Delta_x\psi^\eps + V_\rmb(x)\psi^\eps\,,\\
\psi^\eps(T,x) = \psi_T^\eps(x)\,,
\end{aligned}
\end{equation}
where the data is given at the final time $T$.

Taking $\mbox{\eqref{eqn:schr_perturbed}}\times\overline{\psi^\eps}-\overline{\mbox{\eqref{eqn:schr_conjugate}}}\times\phip$, we arrive at
\begin{equation*}
\ri\eps\partial_t(\phip\overline{\psi^\eps}) = -\half\eps^2(\overline{\psi^\eps}\Delta_x\phip - \phip\Delta_x\overline{\psi^\eps}) + \Vp(x)\phi_\rmb^\eps \overline{\psi^\eps}\,.
\end{equation*}
We integrate the equation over $\Rb^d\times[0,T]$, and then apply the Green's identity and make use of the trivial initial data for $\phip$. This finally yields our problem formulation
\begin{equation}\label{eqn:schr_integral}
\int_{\Rb^d}\Bp\overline{\psi^\eps_T}\rmd x = \frac{1}{\ri\eps}\int_{\Rb^d}\Vp(x) \int_0^T\phi_\rmb^\eps\overline{\psi^\eps} \rmd t\rmd x= \int_{\Rb^d}\Vp(x) R_\rmS^\eps[\phi_{\rmI}^\eps,\psi_T^\eps](x)\rmd x \,,
\end{equation}
where we call the representative:
\begin{equation}\label{eqn:schr_rep}
R_\rmS^\eps[\phi_{\rmI}^\eps,\psi_T^\eps] = \frac{1}{\ri\eps}\int_0^T\phi_\rmb^\eps\overline{\psi^\eps} \rmd t \,.
\end{equation}
Note that the left hand side of~\eqref{eqn:schr_integral} is known, with $\psi_T^\eps$ given in~\eqref{eqn:schr_conjugate} and $\phip_T$ calculated from the measured data~\eqref{eqn:schr_measured}. The right hand side formulates a Fredholm integral on the unknown $\Vp$ and the kernel $R^\eps_\rmS$. Reconstruction of $\Vp$ amounts to inverting this Fredholm integral using different configurations of $R_\rmS^\eps$, which, in turn, are tuned by $(\phi_{\rmI}^\eps\,,\psi^\eps_T)$ data pairs.

\subsection{A linearized inverse problem for the Wigner equation}\label{sec:lin_inv_Wigner}
The counterpart of the Schr\"odinger equation on the phase space is the Wigner equation. We derive the linearized inverse problem for this equation assuming initial and final states are given. Recall the Wigner equation in $\Rb^{2d}$:
\begin{equation}\label{eqn:wigner_full}
\partial_t f^\eps + k\cdot\nabla_x f^\eps = \Lc_{V}^\eps[f^\eps]\,,
\end{equation}
with
\begin{equation*}
\Lc_{V}^\eps[f^\eps] = \ri\int_{\Rb^{2d}} e^{\ri p(x-y)} V(y) \frac{1}{\eps}\left[f^\eps\left(x,k+\half\eps p\right)-f^\eps\left(x,k-\half\eps p\right)\right] \rmd p \rmd y\,.
\end{equation*}

Let the initial data be $f^\eps(0,x,k) = f_\rmI^\eps(x,k)$, and we call the final time data $f^\eps(T,x,k) = f_T^\eps(x,k)$. The goal is to use initial-final data pairs $(f_\rmI^\eps\,,f^\eps_T)$ to reconstruct the potential term $V$. This amounts to using the following operator to reconstruct $V$:
\begin{equation*}
\mathcal{M}^\eps_\rmW[V]\,:\;f_\rmI^\eps\to f^\eps_T\,.
\end{equation*}

\begin{remark}
According to the definition of $\Lc_V$, it is immediate that $\Lc_{V}^\eps = \Lc_{V+C}^\eps$ where $C$ can be any constant. This makes $\mathcal{M}^\eps_\rmW[V] = \mathcal{M}^\eps_\rmW[V+C]$. Therefore the reconstruction can be at most unique up to an unknown constant.
\end{remark}

To derive the linear inverse problem, we assume that there is a background potential $V_\rmb(x)$ so that $\Vp(x) = V(x)-V_\rmb(x)$ is much smaller than $V_\rmb(x)$ in amplitude. Call the background problem with the same initial condition:
\begin{equation}\label{eqn:wigner_background}
\begin{aligned}
\partial_t f_\rmb^\eps + k\cdot\nabla_x f_\rmb^\eps = \Lc_{V_\rmb}^\eps[f_\rmb^\eps]\,,\\
f_\rmb^\eps(0,x,k) = f_\rmI^\eps(x,k)\,,
\end{aligned}
\end{equation}
where the operator $\Lc_{V_\rmb}$ is defined by the background potential. With a preset $V_\rmb$ and $f_\rmI^\eps$, $f^\eps_{\rmb,T}(x,k)=f_\rmb^\eps(T,x,k)$ can be pre-computed.

Define $\fp = f^\eps - f_\rmb^\eps$, we subtract~\eqref{eqn:wigner_background} from~\eqref{eqn:wigner_full}, and drop the higher term $\Lc_{\widetilde{V}}^\eps[\fp]$ to have the equation for $\fp$:
\begin{equation}\label{eqn:wigner_perturbed}
\begin{aligned}
\partial_t \fp + k\cdot\nabla_x \fp = \Lc_{V_\rmb}^\eps[\fp] + \Lc_{\widetilde{V}}^\eps[f_b^\eps]\,,\\
\fp(0,x,k) = 0\,.
\end{aligned}
\end{equation}

This equation describes the dynamics of the perturbed data $\fp$. It has trivial initial data, and implicitly depends on $f^\eps_\rmI$ through the $\Lc_{\widetilde{V}}^\eps[f_\rmb^\eps]$ term. Since $f^\eps_T(x,k)$ is the measured data and $f_{\rmb,T}^\eps(x,k)$ is precomputed, the perturbed equation's final data is also known:
\begin{equation*}
\fp_T(x,k) = \fp(T,x,k)=f^\eps(T,x,k) - f_\rmb^\eps(T,x,k) = f^\eps_T(x,k) - f^\eps_{\rmb,T}(x,k)\,.
\end{equation*}

As was done in the case of the Schr\"odinger equation, we also derive the adjoint equation for $g^\eps$:
\begin{equation}\label{eqn:wigner_adjoint}
\begin{aligned}
\partial_t g^\eps + k\cdot\nabla_x g^\eps = \Lc_{V_\rmb}^\eps[g^\eps]\,,\\
g^\eps(T,x,k) = g_T^\eps(x,k)\,,
\end{aligned}
\end{equation}
with the data imposed at the final time $t=T$.

Taking $\mbox{\eqref{eqn:wigner_perturbed}}\times \overline{g^\eps}+\overline{\mbox{\eqref{eqn:wigner_adjoint}}}\times\fp$, we arrive at
\begin{equation}
\partial_t (\fp\overline{g^\eps}) + \nabla_x\cdot(k\fp\overline{g^\eps}) = \overline{g^\eps} \Lc_{V_\rmb}^\eps[\fp] + \fp\overline{\Lc_{V_\rmb}^\eps[g^\eps]} + \overline{g^\eps} \Lc_{\widetilde{V}}^\eps[f_\rmb^\eps] \,.
\end{equation}
We integrate the equation over $\Rb^{2d}\times[0,T]$. Making use of the anti-self-adjointness of $\Lc_{V_\rmb}^\eps$, as shown in~\eqref{eqn:self_ajoint_L}, and the trivial initial data of $\fp$, we obtain:
\begin{equation*}\label{eqn:wigner_integral1}
\int_{\Rb^{2d}}\widetilde{f}_T^\eps \overline{g_T^\eps}\rmd x \rmd k
= \int_{\Rb^{2d}\times[0,T]}\overline{g^\eps} \Lc_{\Vp}^\eps[f_\rmb^\eps]\rmd x \rmd k \rmd t\,.
\end{equation*}

We note that the integral term on the right hand side of the equation is a linear operator on $\Vp$. To do so, we expand $\Lc_{\Vp}^\eps$, and employ~\eqref{eqn:col_Deps}:
\begin{equation}\label{eqn:wigner_integral}
\int_{\Rb^{2d}}\widetilde{f}_T^\eps \overline{g_T^\eps}\rmd x \rmd k = \int_{\Rb^{d}} \Vp(x)R_\rmW^\eps[f_{\rmI}^\eps,g_T^\eps](x)\rmd x\,,
\end{equation}
with the representative
\begin{equation}\label{eqn:wigner_rep}
R_\rmW^\eps[f_{\rmI}^\eps,g_T^\eps] = \frac{\ri}{(2\pi)^d} \int_{\Rb^{3d}\times[0,T]} e^{\ri p (z-x)} \overline{g^\eps}(z,k) D^\eps f_\rmb^\eps(z,k,p) \rmd p \rmd z \rmd k \rmd t\,,
\end{equation}
where $D^\eps f_\rmb^\eps$ is defined in~\eqref{eqn:D_eps}.

Once again, the left hand side of~\eqref{eqn:wigner_integral} is known, with $\fp_T$ computed and $g_T^\eps$ given, and the right hand side of~\eqref{eqn:wigner_integral} is a Fredholm integral on $\Vp$ with the kernel $R_\rmW^\eps$. The linear inverse problem of the Wigner equation amounts to inverting such an integral. By choosing different configurations of $(f_{\rmI}^\eps,g_T^\eps)$, we obtain different profiles of $R^\eps_\rmW$, using which, we try to reconstruct $\Vp$.

\subsection{A linearized inverse problem for the Liouville equation}\label{sec:lin_inv_Liouville}
Finally, we derive the inverse problem for the Liouville equation. Recall the Liouville equation in $\Rb^{2d}$
\begin{equation}\label{eqn:liouville_full}
\partial_t f + k\cdot\nabla_x f - \nabla_x V \cdot \nabla_k f = 0\,.
\end{equation}
Denote the initial data to be $f(0,x,k) = f_\rmI(x,k)$, and we assume one can experimentally measure the final time solution at $t=T$ for $f_T(x,k)=f(T,x,k)$. The goal is to reconstruct $V$ in the Liouville equation using the initial-to-final data pairs $(f_\rmI\,,f_T)$. Namely, to use the following operator to reconstruct $V$:
\begin{equation*}
\mathcal{M}^\eps_\rmL[V]\,:\;f_\rmI \to f_T\,.
\end{equation*}
\begin{remark}
Since the potential term enters the equation~\eqref{eqn:liouville_full} through its gradient $\nabla_x V$, $\mathcal{M}^\eps_\rmL[V] = \mathcal{M}^\eps_\rmL[V+C]$ for any constant $C$. This means the reconstruction of $V$ is at most unique up to an arbitrary constant.
\end{remark}

As was done in the previous sections, the problem shall be linearized around a background potential $V_\rmb(x)$. The background equation with the same initial data writes:
\begin{equation}\label{eqn:liouville_background}
\begin{aligned}
\partial_t f_\rmb + k\cdot\nabla_x f_\rmb - \nabla_x V_\rmb \cdot \nabla_k f_\rmb = 0\,,\\
f_\rmb(0,x,k) = f_\rmI(x,k)\,.
\end{aligned}
\end{equation}
Denoting the perturbation $\widetilde{f} = f - f_\rmb$, we subtract~\eqref{eqn:liouville_background} from~\eqref{eqn:liouville_full}, and drop the higher order term $\nabla_x \Vp \cdot \nabla_k \widetilde{f}$ to obtain the equation for the perturbation:
\begin{equation}\label{eqn:liouville_perturbed}
\begin{aligned}
\partial_t \widetilde{f} + k\cdot\nabla_x \widetilde{f} - \nabla_x V_\rmb \cdot \nabla_k \widetilde{f} = \nabla_x \widetilde{V} \cdot \nabla_k f_\rmb\,,\\
\widetilde{f}(0,x,k) = 0\,.
\end{aligned}
\end{equation}
The equation has trivial initial data, but it implicitly depends on $f_\rmI$ through the $f_\rmb$ term that enters as the source. Since $f_T(x,k)$ is measured, and $f_\rmb(T,x,k)$ can be precomputed for any given $V_\rmb$, we easily obtain:
\begin{equation*}
\widetilde{f}_T = \widetilde{f}(T,x,k) =  f(T,x,k) - f_\rmb(T,x,k)
\end{equation*}
as a known quantity.

The adjoint equation for $g$ is:
\begin{equation}\label{eqn:liouville_adjoint}
\begin{aligned}
\partial_t g + k\cdot\nabla_x g - \nabla_x V_\rmb \cdot \nabla_k g = 0\,,\\
g(T,x,k) = g_T(x,k)\,.
\end{aligned}
\end{equation}

Taking $\mbox{\eqref{eqn:liouville_perturbed}}\times \overline{g}+\overline{\mbox{\eqref{eqn:liouville_adjoint}}}\times\widetilde{f}$, we arrive at
\begin{equation*}
\partial_t (\widetilde{f}\overline{g}) + \nabla_x\cdot (k\widetilde{f}\overline{g}) -   \nabla_k \cdot (\widetilde{f}\overline{g} \nabla_x V_\rmb) = \overline{g}\nabla_x \widetilde{V} \cdot \nabla_k f_\rmb \,.
\end{equation*}
We integrate the equation over $\Rb^{2d}\times[0,T]$, and make use of the trivial initial data for $\widetilde{f}$:
\begin{equation*}
\int_{\Rb^{2d}}\widetilde{f}_T \overline{g}_T\rmd x \rmd k
= \int_{\Rb^{2d}\times[0,T]}\overline{g} \nabla_x\widetilde{V}\cdot\nabla_k f_b\rmd x \rmd k \rmd t\,.
\end{equation*}
Moving the $\nabla_x$ from $V$ to $\overline{g}\nabla_k f_\rmb$, this becomes
\begin{equation}\label{eqn:liouville_integral}
\int_{\Rb^{2d}}\widetilde{f}_T \overline{g}_T\rmd x \rmd k = \int_{\Rb^d}\widetilde{V}(x) R_\rmL[f_{\rmI},g_T](x)\rmd x\,,
\end{equation}
where the Liouville representative is defined as:
\begin{equation}\label{eqn:liouville_rep}
R_\rmL[f_{\rmI},g_T] = -\nabla_x\cdot\int_{\Rb^d\times[0,T]}\overline{g}\nabla_k f_\rmb \rmd k \rmd t \,.
\end{equation}
Again, the left hand side of~\eqref{eqn:liouville_integral} is known, and the right hand side of~\eqref{eqn:liouville_integral} is a Fredholm integral of $\Vp$ with the kernel $R_\rmL$. The linear inverse problem of the Liouville equation is to invert such an integral.

\section{Connecting the three inverse problems}\label{sec:connecting}
The Schr\"odinger equation, the Wigner equation and the Liouville equation are connected. According to Lemma~\ref{lem:wigner} and Theorem~\ref{thm:semiclassical_limit}, $f^\eps = W^\eps[\phi^\eps]$ necessarily satisfies the Wigner equation as long as $\phi^\eps$ solves the Schr\"odinger equation, and when $\eps\to0$, $f^\eps\to f$ that solves the Liouville equation.

We look for the counterparts of these relations in the inverse setting. This is to investigate the three inverse problems introduced in Section~\ref{sec:inverse}. More specifically, since the three inverse problems are uniquely represented by the three representatives $R_\rmS^\eps$, $R_\rmW^\eps$ and $R_\rmL$, as defined in~\eqref{eqn:schr_rep},~\eqref{eqn:wigner_rep} and~\eqref{eqn:liouville_rep} respectively, we essentially need to show the connections between them.

\subsection{From Schr\"odinger to Wigner in the inverse setting}\label{sec:schr_to_wigner}
This is to study the relation between the linear Schr\"odinger inverse problem~\eqref{eqn:schr_integral} and the linear Wigner inverse problem~\eqref{eqn:wigner_integral}. For simplicity of notations, we drop the $\eps$ superscript throughout the subsection. The theorem below demonstrates that every Wigner representative $R_\rmW$ can be written as a linear combination of Schr\"odinger representatives $R_\rmS$. This means the space spanned by all $R_\rmW$ is a subspace spanned by all $R_\rmS$.

\begin{theorem}
Let $\phi_\rmb(t)$ and $\phi_\rmb'(t)$ be the solutions to the background Schr\"odinger equation~\eqref{eqn:schr_background} with initial data $\phi_\rmI$ and $\phi_\rmI'$ respectively. Let $\psi(t)$ and $\psi'(t)$ be the solutions to the adjoint Sch\"odinger equation~\eqref{eqn:schr_conjugate} with final data $\psi_T$ and $\psi_T'$ respectively. More over, let $f_\rmI = W[\phi_\rmI,\psi'(0)]$ and $g_T = W[\psi_T,\phi_\rmb'(T)]$. Then

\begin{equation}\label{eqn:wig_rep_to_schr_rep}
  (2\pi\eps)^d R_\rmW[f_\rmI,g_T]
=  \inprod{\phi_\rmI'}{\psi'(0)}{} R_\rmS[\phi_\rmI,\psi_T]
- \inprod{\phi_\rmI}{\psi(0)}{} R_\rmS[\phi_\rmI',\psi_T'] \,.
\end{equation}

\end{theorem}
\begin{proof}
According to the definition of Wigner representative~\eqref{eqn:wigner_rep}, let $f_\rmb$ and $g$ solve the background and the adjoint Wigner equations, $R_\rmW[f_\rmI,g_T]$ becomes
\begin{equation}\label{eqn:R_W_I}
R_\rmW[f_\rmI,g_T] = \frac{\ri}{(2\pi)^d}\int_{0}^T I(t)\rmd t
\end{equation}
with
\begin{equation*}
I = \int_{\Rb^{3d}} e^{\ri p (z-x)} \overline{g}(z,k)  D f_\rmb(z,k,p) \rmd p \rmd z \rmd k\,.
\end{equation*}

According to~\eqref{lem:wigner}, $g(t) = W[\psi(t),\phi_\rmb'(t)]$ and $f_\rmb(t)=W[\phi_\rmb(t),\psi'(t)]$ and thus
\begin{equation*}
I =\int_{\Rb^{3d}} e^{\ri p (z-x)} \overline{W[\psi(t),\phi_\rmb'(t)]}(z,k)  D W[\phi_\rmb(t),\psi'(t)](z,k,p) \rmd p \rmd z \rmd k\,.
\end{equation*}
Here $D W[\phi_\rmb(t),\psi'(t)]$ is defined as in~\eqref{eqn:D_eps}. Plugging in the Wigner transform, we get
\begin{equation}
\begin{aligned}
I= &\frac{1}{(2\pi)^{2d}\eps} \int_{\Rb^{5d}}  e^{\ri k (q-y)} [e^{\ri p \left(z+\half\eps q -x\right)} - e^{\ri p\left( z-\half\eps q -x\right)}] \\
&\overline{\psi}\left(z-\half\eps y\right) \phi_\rmb'\left(z+\half\eps y\right) \phi_\rmb\left(z-\half\eps q\right) \overline{\psi'}\left(z+\half\eps q\right) \rmd y \rmd q \rmd p \rmd z \rmd k \\
= &\frac{1}{(2\pi)^d\eps} \int_{\Rb^{3d}}  [e^{\ri p \left(z+\half\eps q -x\right)} - e^{\ri p\left( z-\half\eps q -x\right)}] \\
&\overline{\psi}\left(z-\half\eps q\right) \phi_\rmb'\left(z+\half\eps q\right) \phi_\rmb\left(z-\half\eps q\right) \overline{\psi'}\left(z+\half\eps q\right) \rmd q \rmd p \rmd z \,,
\end{aligned}
\end{equation}
where we used the Fourier inversion formula
\begin{equation}\label{eqn:fourier_inv}
\frac{1}{(2\pi)^d} \int_{\Rb^{2d}} e^{\ri k (q-y)} h(y) \rmd y \rmd k = h(q)\,.
\end{equation}

Let $z' = z+\half\eps q, z'' = z-\half\eps q$, we get
\begin{equation}
\begin{aligned}
I&= \frac{1}{(2\pi)^d\eps^{d+1}} \int_{\Rb^{3d}}  [e^{\ri p (z' -x)} - e^{\ri p( z'' -x)}]
\overline{\psi}(z'') \phi_\rmb'(z') \phi_\rmb(z'') \overline{\psi'}(z') \rmd z' \rmd p \rmd z'' \\
&=\frac{1}{\eps^{d+1}}  [\inprod{\phi_\rmb(t)}{\psi(t)}{\LdC}   \phi_\rmb'(x)\overline{\psi'}(x) -  \inprod{\phi_\rmb'(t)}{\psi'(t)}{\LdC} \phi_\rmb(x)\overline{\psi}(x)]
\end{aligned}
\end{equation}
where we again use the Fourier inversion formula in the second equality.

According to~\eqref{eqn:inprod_conserve}, $\inprod{\phi_\rmb(t)}{\psi(t)}{}$ and $\inprod{\phi_\rmb'(t)}{\psi'(t)}{}$ are both constants independent of $t$. Using~\eqref{eqn:R_W_I} and the definition of the Schr\"odinger representative~\eqref{eqn:schr_rep}, we integrate the equation over $[0,T]$ to conclude~\eqref{eqn:wig_rep_to_schr_rep}.
\end{proof}

The unique reconstruction of $\Vp$ in~\eqref{eqn:schr_integral} and~\eqref{eqn:wigner_integral} amounts to investigating the dimension of the spaces $R_\rmS$ and $R_\rmW$ respectively. This theorem suggests that the latter space is a subspace of the former, meaning the unique reconstruction of~\eqref{eqn:wigner_integral} would indicate the unique reconstruction of~\eqref{eqn:schr_integral}.

\subsection{From Wigner to Liouville in the inverse setting}\label{sec:wigner_to_liouville}
According to Theorem~\ref{thm:semiclassical_limit}, the Liouville equation is the classical limit of the Wigner equation, meaning that $f^\eps$, which solves the Wigner equation, converges to $f$, which solves the Liouville equation, when $\eps\to 0$. We expect similar argument holds true in the inverse setting as well. This amounts to study the two representatives $R^\eps_\rmW$ and $R_\rmL$.
\begin{theorem}
Let $R_\rmW^\eps[f_{ \rmI}^\eps,g_T^\eps]$, and $R_\rmL[f_{ \rmI},g_T]$ be the representatives defined in~\eqref{eqn:wigner_rep} and~\eqref{eqn:liouville_rep} respectively, where
\begin{equation*}
f_{ \rmI} = \lim_{\eps\to0}f^\eps_{ \rmI}\,,\quad g_T = \lim_{\eps\to0}g^\eps_T\,,
\end{equation*}
then we claim:
\begin{equation}\label{eqn:conv_wigner_liouville}
\lim_{\eps\to0}R_\rmW^\eps[f_{ \rmI}^\eps,g_T^\eps] = R_\rmL[f_{ \rmI},g_T]\,.
\end{equation}
\end{theorem}
\begin{proof}
Suppose $f_\rmb^\eps$ solves the background Wigner equation~\eqref{eqn:wigner_background} with the initial data $f^\eps_{\rmI}$, and $f_\rmb$ solves the background Liouville equation~\eqref{eqn:liouville_background} with the initial data $f_{\rmI}$. In the semi-classical regime $\eps\rightarrow0$, by Theorem~\ref{thm:semiclassical_limit}, we know that the background wave $f_\rmb^\eps$ converges to $f_\rmb$. Thus, we have, formally,
\begin{equation}\label{eqn:conv_Df}
\begin{aligned}
\lim_{\eps\to0}D^\eps f_\rmb^\eps(z,k,p) &= \lim_{\eps\to0}\frac{1}{\eps}\left[f_\rmb^\eps\left(z,k+\half\eps p \right) - f_\rmb^\eps\left(z,k-\half\eps p \right) \right] \\
&\overset{\eps\rightarrow0}{\rightarrow} p\cdot\nabla_k f_\rmb(z,k)\,.
\end{aligned}
\end{equation}
Similarly, suppose $g^\eps$ solves the adjoint Wigner equation~\eqref{eqn:wigner_adjoint} with the final data $g_T^\eps$, and $g$ solves the adjoint Liouville equation~\eqref{eqn:liouville_adjoint} with the final data $g_T$, then the adjoint wave $g^\eps$ converges to $g$. Combining these:
\begin{equation}
\begin{aligned}
R_\rmW^\eps[f_{ \rmI}^\eps,g_T^\eps](x) &= \frac{\ri}{(2\pi)^d} \int_{\Rb^{3d}\times[0,T]} e^{\ri p (z-x)} \overline{g^\eps}(z,k) D^\eps f_\rmb^\eps(z,k,p) \rmd p \rmd z \rmd k \rmd t \\
&\overset{\eps\rightarrow0}{\rightarrow} \frac{\ri}{(2\pi)^d} \int_{\Rb^{3d}\times[0,T]} e^{\ri p (z-x)} \overline{g}(z,k) p\cdot\nabla_k f_\rmb(z,k) \rmd p \rmd z \rmd k \rmd t \,.
\end{aligned}
\end{equation}
Integrating by parts for the limit and applying the Fourier inversion formula~\eqref{eqn:fourier_inv} lead to the Liouville representative, that is, the right hand side of the last limit becomes:
\begin{equation}
\begin{aligned}
&- \frac{1}{(2\pi)^d} \int_{\Rb^{3d}\times[0,T]} e^{\ri p (z-x)} \nabla_z \cdot (\overline{g^\eps}(z,k) \nabla_k f_\rmb(z,k)) \rmd p \rmd z \rmd k \rmd t \\
=& - \nabla_x\cdot\int_{\Rb^d\times[0,T]} \overline{g}(x,k) \nabla_k f_\rmb(x,k) \rmd k \rmd t =  R_\rmL[f_{\rmI},g_T] \,,
\end{aligned}
\end{equation}
which concludes~\eqref{eqn:conv_wigner_liouville}.
\end{proof}

The proof above is formal. We assumed enough regularity for the convergence~\eqref{eqn:conv_Df}. We also need the convergence to hold true in the strong sense (in $L^2$ for example).

This theorem suggests that the inverse problem of the Wigner equation, in the classical regime with $\eps\to0$, is asymptotically equivalent to the inverse problem of the Liouville equation. This connects the Schr\"odinger equation with the Newton's law of motion in the inverse setting: if the reconstruction of the potential term using the initial-to-final data map is unique (up to a gauge transform) and stable for the Schr\"odinger equation, the same holds true for the Newton's law of motion.

\section{Numerical results}\label{sec:num_test}

As a proof of concept, we provide numerical evidences for the Wigner inverse problem~\eqref{eqn:wigner_integral} in the classical regime.

\subsection{Numerical setup}
In $(1+1)$-dimensional space, the background Wigner equation reads
\begin{equation}\label{eqn:wigner_background_1d}
\begin{aligned}
&\partial_t f_\rmb^\eps + k\partial_x f_\rmb^\eps =
 \frac{\ri}{2\pi\eps} \int_{\Rb^2}\left[ V_\rmb\left(x+\frac{\eps y}{2}\right) - V_\rmb\left(x-\frac{\eps y}{2}\right) \right] f_\rmb^\eps (x,p) e^{\ri y(k-p)}   \rmd y \rmd p\,, \\
&f_\rmb^\eps(0,x,k) = f_{\rmI}^\eps(x,k)\,,
\end{aligned}
\end{equation}
and its adjoint equation reads
\begin{equation}\label{eqn:wigner_adjoint_1d}
\begin{aligned}
&\partial_t g^\eps + k\partial_x g^\eps =
 \frac{\ri}{2\pi\eps} \int_{\Rb^2}\left[ V_\rmb\left(x+\frac{\eps y}{2}\right) - V_\rmb\left(x-\frac{\eps y}{2}\right) \right] g(x,p) e^{\ri y(k-p)}   \rmd y \rmd p\,, \\
&g^\eps(T,x,k) = g_T^\eps(x,k)\,.
\end{aligned}
\end{equation}
The corresponding Liouville equation for the background equation and the adjoint equation are
\begin{equation}\label{eqn:liouville_background_1d}
\begin{aligned}
&\partial_t f_\rmb + k\partial_x f_\rmb - \partial_x V_\rmb \partial_k f_\rmb = 0 \,, \\
&f_\rmb(0,x,k) = f_{ \rmI}(x,k)\,,
\end{aligned}
\end{equation}
and
\begin{equation}\label{eqn:liouville_adjoint_1d}
\begin{aligned}
&\partial_t g + k\partial_x g - \partial_x V_\rmb \partial_k g = 0 \,, \\
&g(T,x,k) = g_T(x,k)\,,
\end{aligned}
\end{equation}
respectively. According to the definitions~\eqref{eqn:wigner_rep} and~\eqref{eqn:liouville_rep}, the Wigner and Liouville representatives in the inverse problems are
\begin{equation*}
\begin{aligned}
&R_\rmW^\eps[f_{ \rmI}^\eps,g_T^\eps](x) \\
&= \frac{\ri}{2\pi} \int_{\Rb^3\times[0,T]} e^{\ri p (z-x)} g^\eps(z,k) \frac{1}{\eps}\left[f_\rmb^\eps\left(z,k+\frac{\eps p}{2}\right)-f_\rmb^\eps\left(z,k-\frac{\eps p}{2}\right)\right] \rmd p \rmd z \rmd k \rmd t \,,
\end{aligned}
\end{equation*}
and
\begin{equation*}
R_\rmL[f_{\rmI},g_T](x)
= -\partial_x \int_{\Rb^d\times[0,T]}  g(x,k) \partial_k f_\rmb(x,k)  \rmd k \rmd t \,,
\end{equation*}
where $f^\eps_\rmb$, $g^\eps$, $f_\rmb$ and $g$ satisfy the equations above. We are to demonstrate the relation between the two representatives as $\eps\to0$.

To set up the experiment, we choose the background potential $V_\rmb(x)$ to have a Gaussian form
\begin{equation}\label{eqn:V_b_numerics}
V_\rmb(x) = A \exp\left(-\frac{(x-a)^2}{w^2}\right)\,,
\end{equation}
and the initial and final time conditions are
\begin{equation}\label{eqn:f_bI}
f_{ \rmI}^\eps(x,k) = f_{ \rmI}(x,k) = B \exp\left(-\frac{(x-b_x)^2}{\sigma_x^2}-\frac{(k-b_k)^2}{\sigma_k^2}\right)\,,
\end{equation}
and
\begin{equation}\label{eqn:g_T}
g_T^\eps(x,k) = g_T(x,k) = C \exp\left(-\frac{(x-c_x)^2}{\delta_x^2}-\frac{(k-c_k)^2}{\delta_k^2}\right)\,.
\end{equation}

To compute the Wigner equation~\eqref{eqn:wigner_background_1d} and~\eqref{eqn:wigner_adjoint_1d}, we truncate the computational domain to $\Omega=[0,L]\times[K_1,K_2]$ and apply periodic boundary condition on $x$. The time interval is taken to be $[0,T]$. The transport term is discretized by a fifth-order WENO scheme~\cite{JiSh:1996efficient}, and the collision term is computed by the trapezoidal approximate~\cite{CaGaMaSh:2003}.

To compute the Liouville equation~\eqref{eqn:liouville_background_1d} and~\eqref{eqn:liouville_adjoint_1d}, we use the particle method. This is to solve the ODE systems of trajectories. For example, to compute~\eqref{eqn:liouville_background_1d} for $0\leq t\leq T$, we set the trajectory equation
\begin{equation*}
\dot{x} = -k\,, \quad
\dot{k} = \partial_x V_\rmb(x)\,,\quad\mbox{with}\quad x(0) = y\,,\; k(0) = p\,,
\end{equation*}
and the initial data for the particle $(y,p)$ is determined by $f_\rmI$. The final solution is thus $f_\rmb(T,y,p) = f_\rmI(x(T),k(T))$.

\subsection{Numerical examples}

In the numerical examples, we set the parameters in~\eqref{eqn:V_b_numerics} to be
\begin{equation*}
A = 1\,, \quad a = 0.25\,,\quad w = 2^{-3}\,,
\end{equation*}
and the parameters defined in~\eqref{eqn:f_bI} and~\eqref{eqn:g_T} are
\begin{equation*}
B = C = 1\,,\quad \sigma_x = \delta_x = 2^{-4}\,,\quad \sigma_k = \delta_k = 2^{-3}\,,\quad b_k = c_k = 2^{-3}\,.
\end{equation*}
For discretization, we use $\Delta x = 2^{-10}$, $\Delta k = 2^{-10}$ and $\Delta t = 2^{-10}$ in both the Wigner and the Liouville solver. In the Wigner solver, we set $L = 0.5$, $K_1 = -0.375$, $K_2 = 0.625$. The terminal time is set to be $T = 2^{-6}$.

In Figure~\ref{fig:f_contour} and Figure~\ref{fig:g_contour}, we first plot the level sets of solutions $f^\eps_\rmb$ and $g^\eps$ at different time.

\begin{figure}[tbhp]
  \centering
  \subfloat[$t = 2^{-7}$]{
  \includegraphics[width=0.45\textwidth,height = 0.18\textheight]{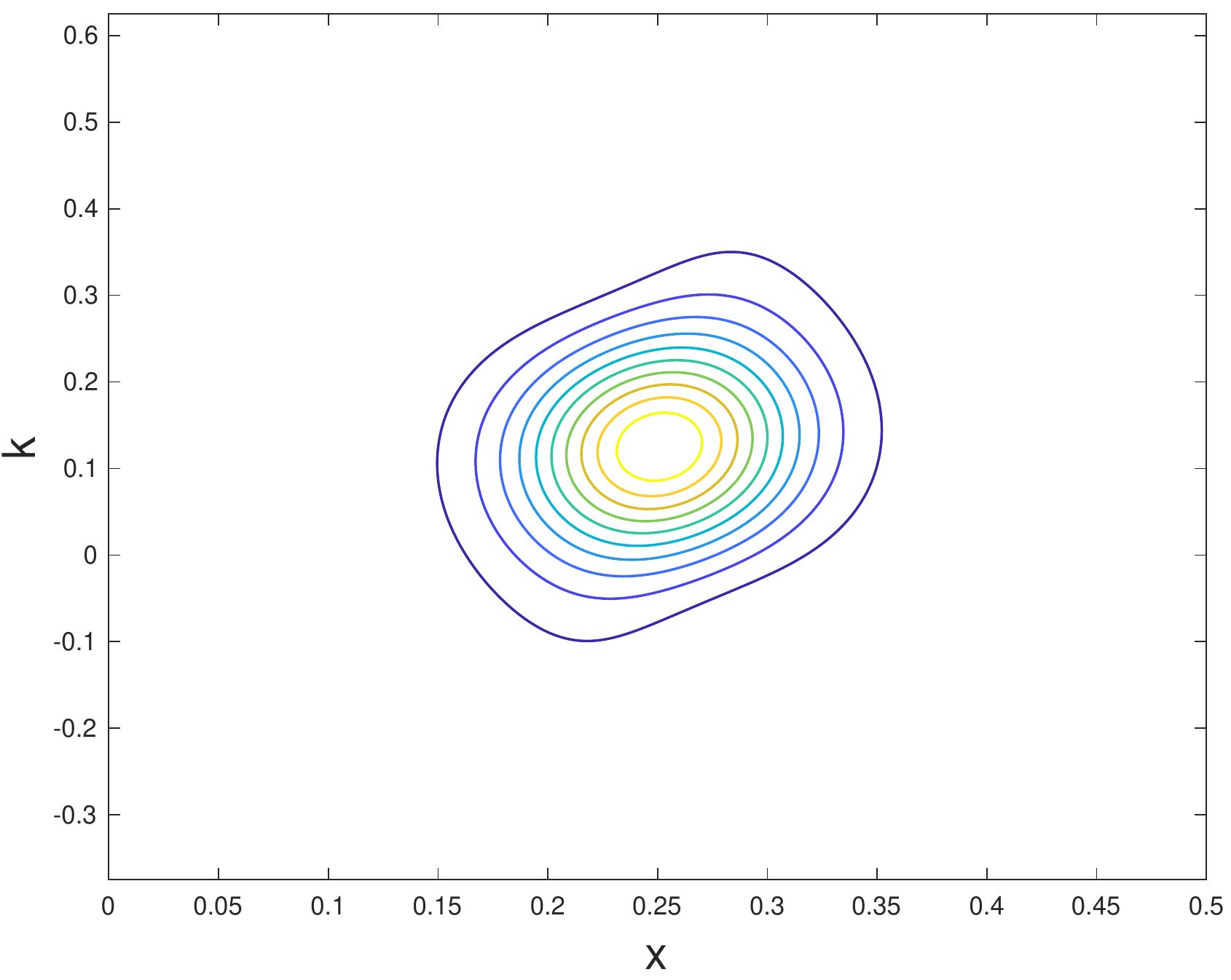}
  \includegraphics[width=0.45\textwidth,height = 0.18\textheight]{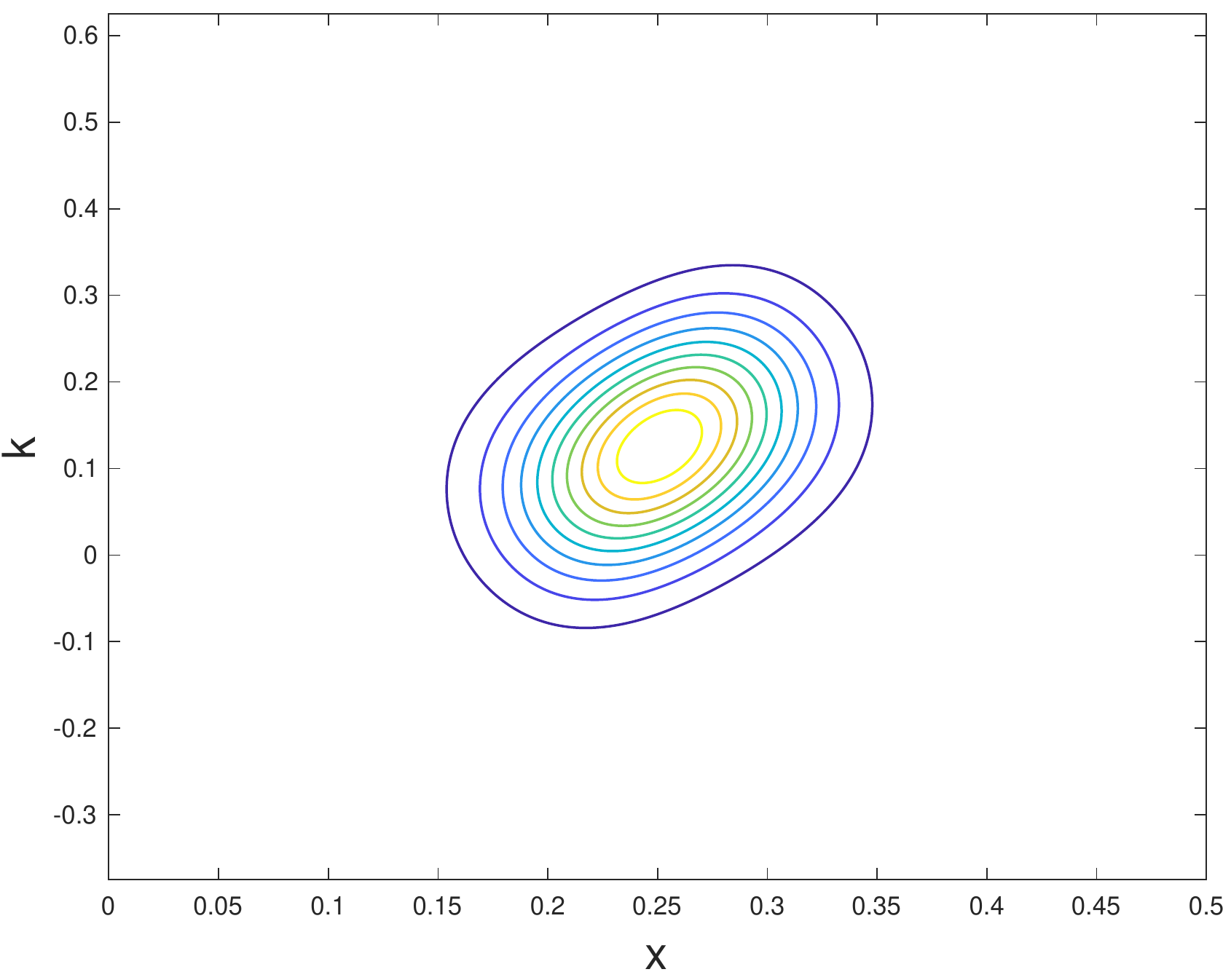}
  }

  \subfloat[$t = 2^{-6}$]{
  \includegraphics[width=0.45\textwidth,height = 0.18\textheight]{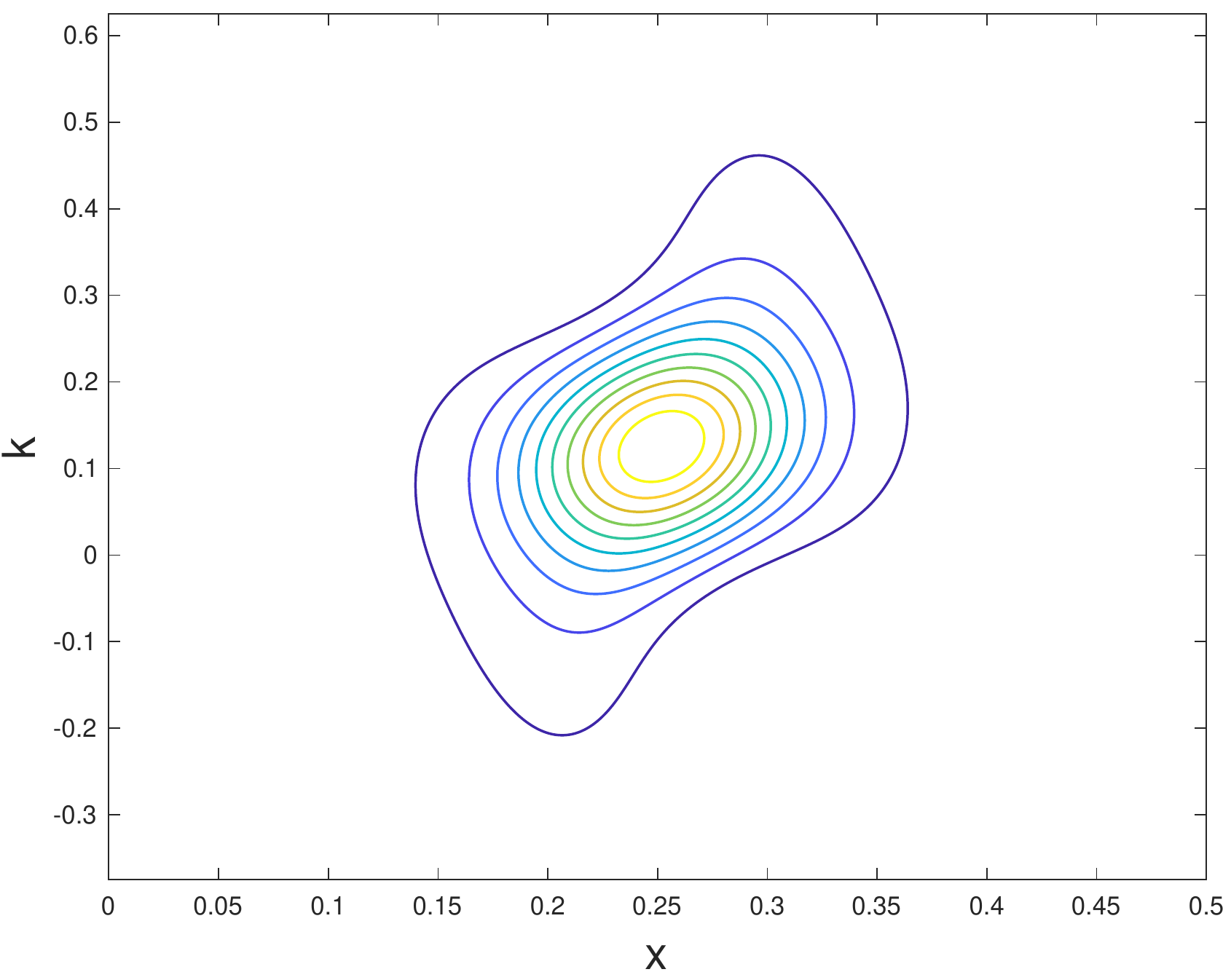}
  \includegraphics[width=0.45\textwidth,height = 0.18\textheight]{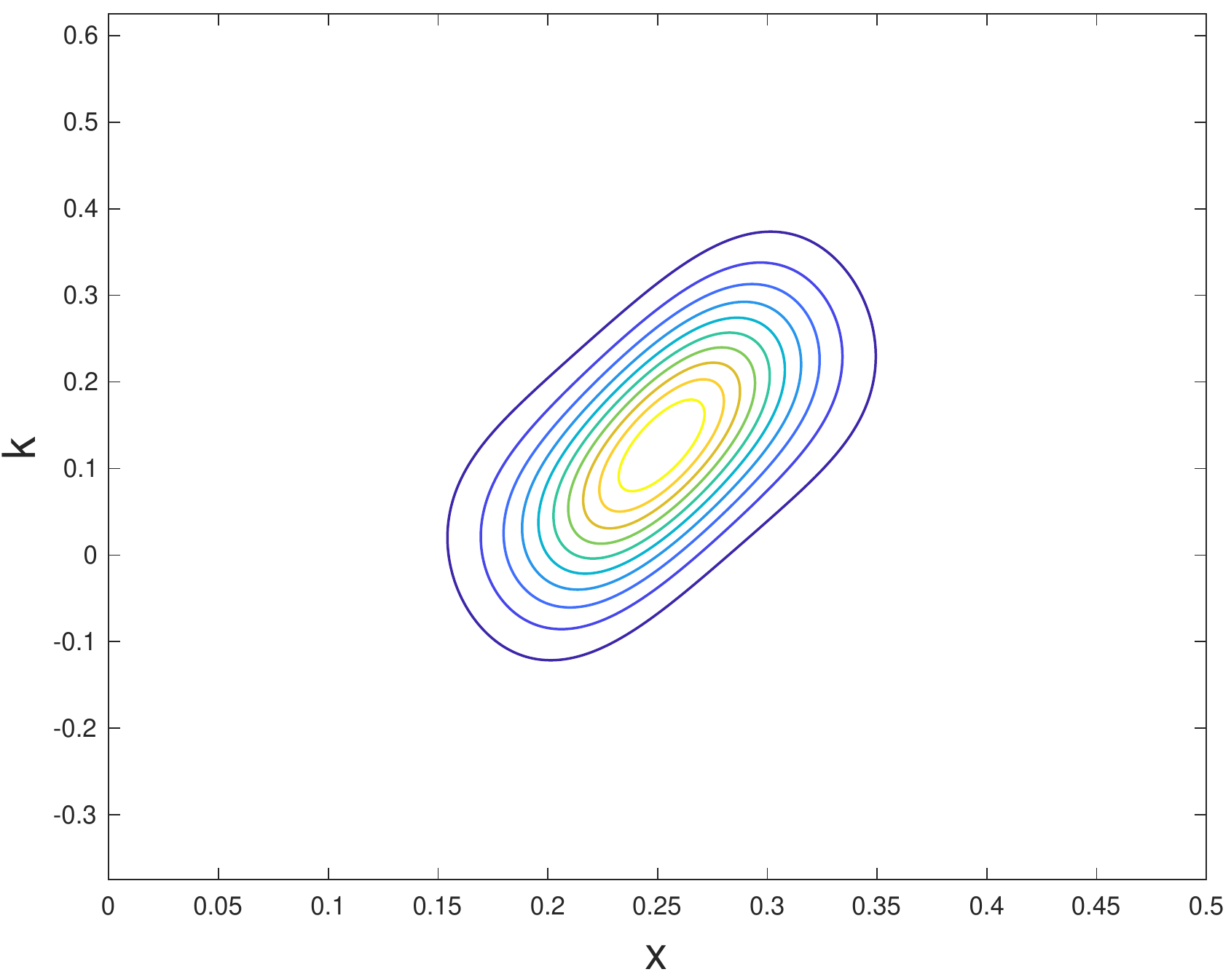}
  }

  \caption{The left column shows the contour of $f_\rmb^\eps$ for $\eps = \pi^{-1}2^{-4}$ and the right column shows the contour for $\eps = \pi^{-1}2^{-8}$.}
  \label{fig:f_contour}
\end{figure}

\begin{figure}[tbhp]
  \centering
  \subfloat[$t = 2^{-7}$]{
  \includegraphics[width=0.45\textwidth,height = 0.18\textheight]{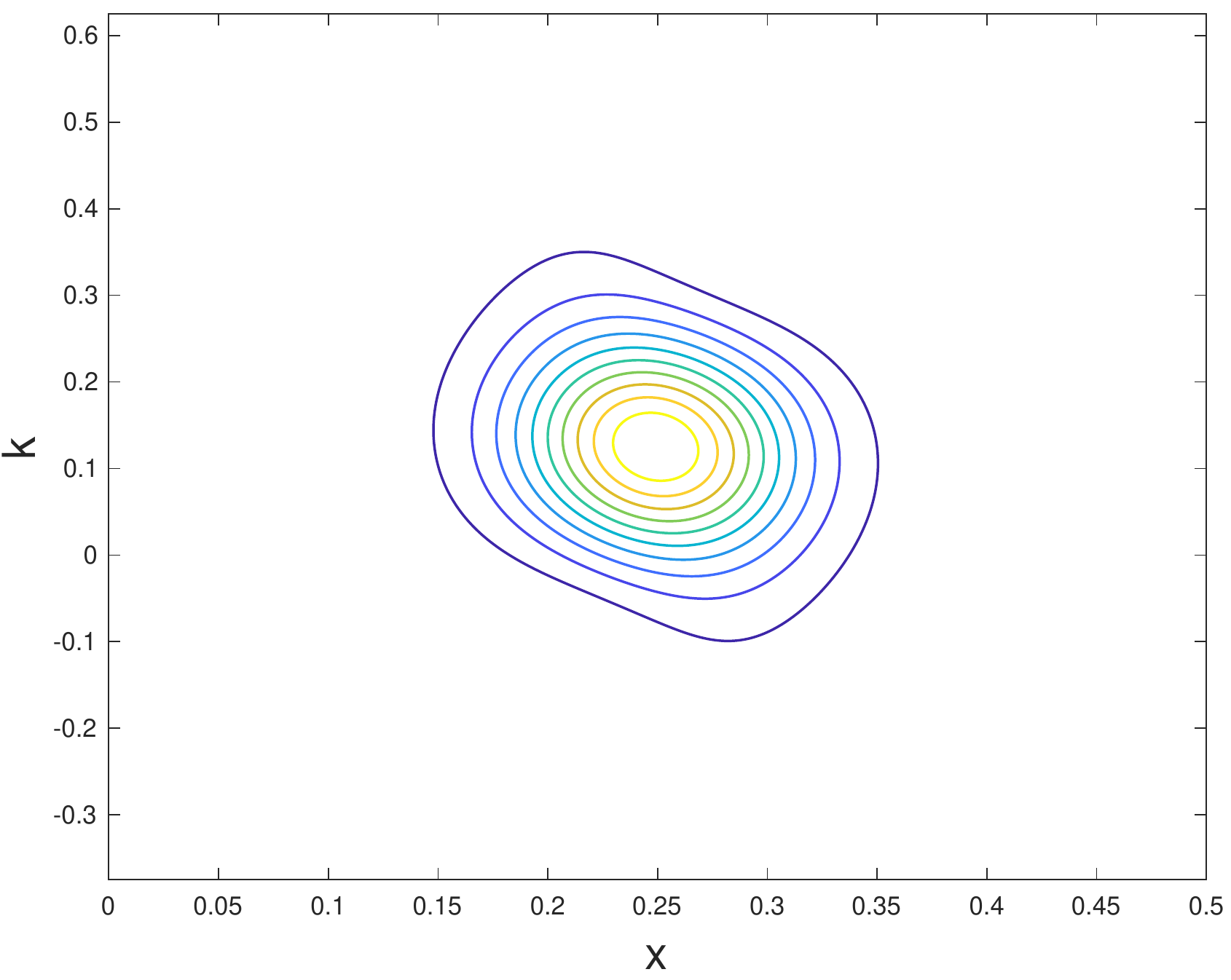}
  \includegraphics[width=0.45\textwidth,height = 0.18\textheight]{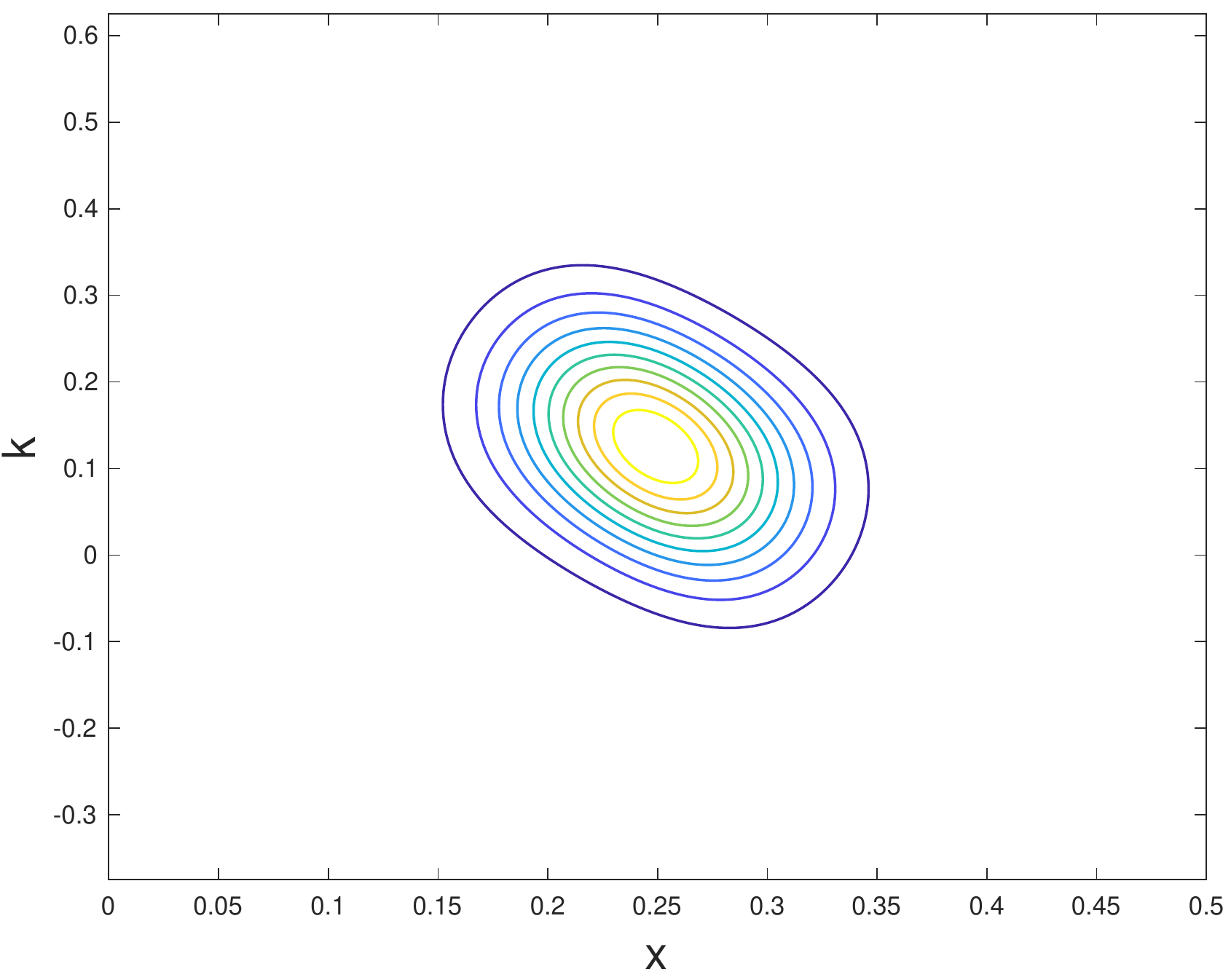}
  }

  \subfloat[$t = 0$]{
  \includegraphics[width=0.45\textwidth,height = 0.18\textheight]{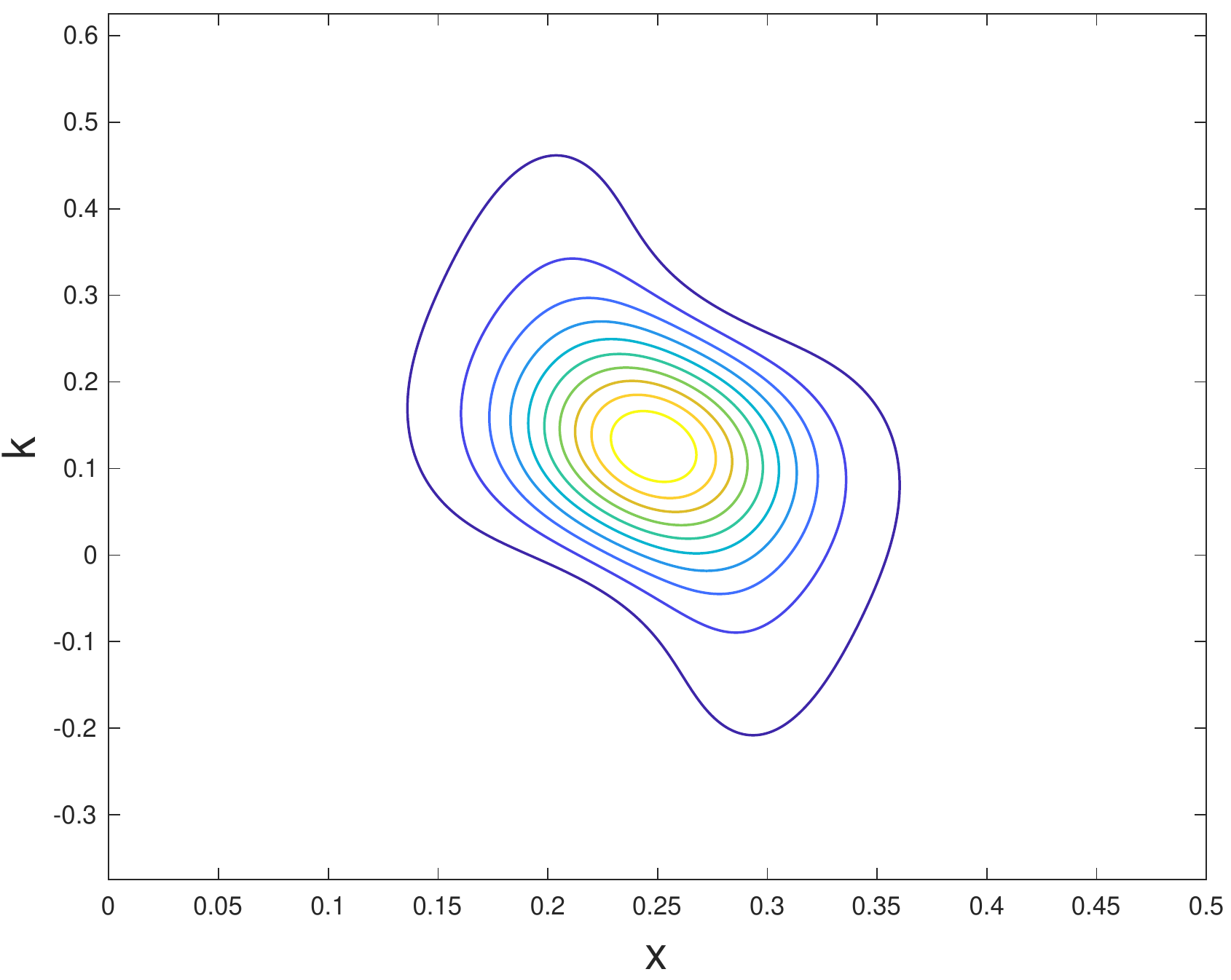}
  \includegraphics[width=0.45\textwidth,height = 0.18\textheight]{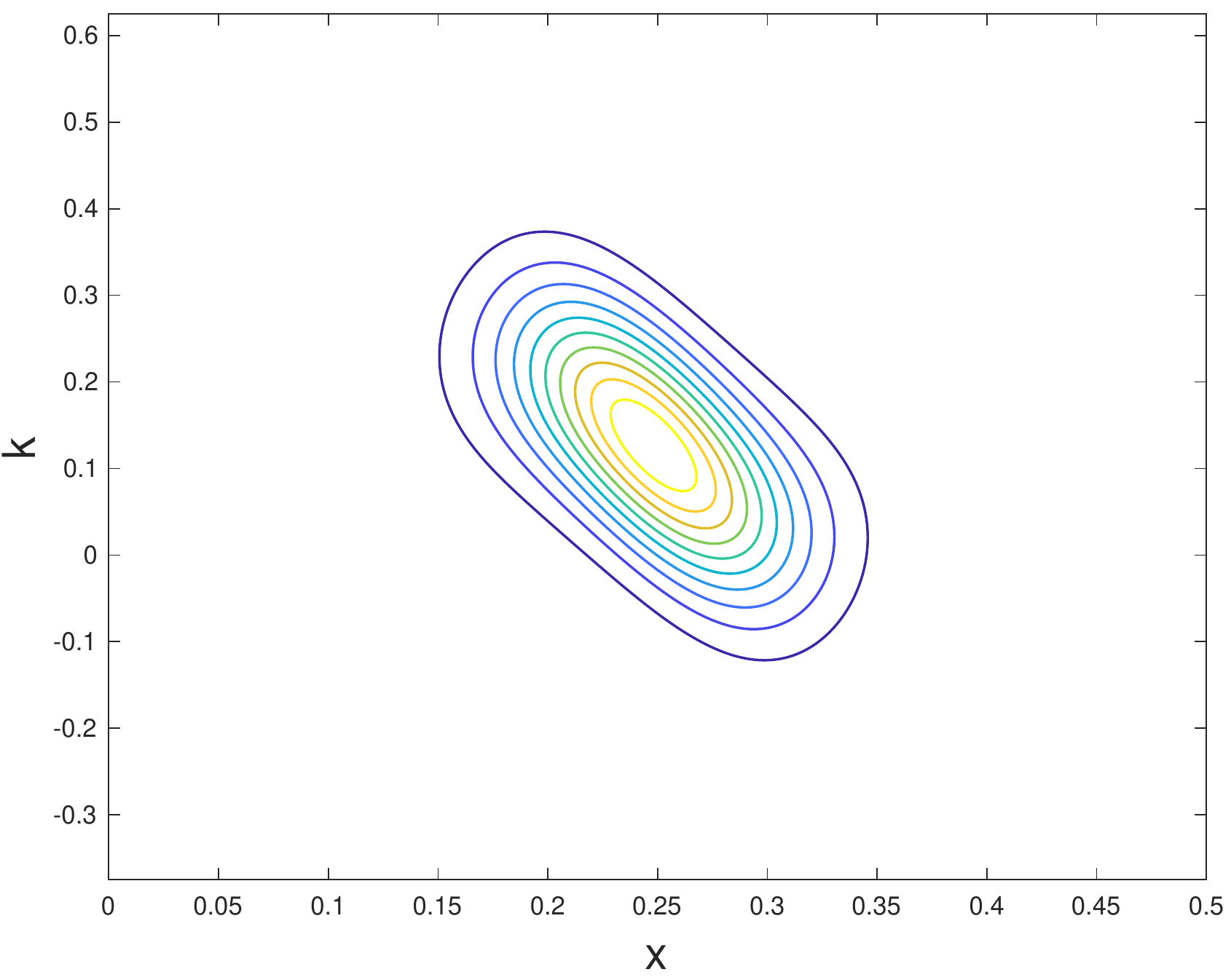}
  }

  \caption{The left column shows the contour of $g^\eps$ for $\eps = \pi^{-1}2^{-4}$ and the right column shows the contour for $\eps = \pi^{-1}2^{-8}$.}
  \label{fig:g_contour}
\end{figure}

We then compare the two representatives $R^\eps_\rmW[f_{ \rmI}^\eps,g_T^\eps]$ and $R_L[f_{ \rmI},g_T]$ for two different configurations of $(b_x,c_x)$. As shown in the left column of Figure~\ref{fig:R_err}, with $\eps\to0$, the profile of $R_\rmW^\eps[f_{ \rmI}^\eps,g_T^\eps]$ gets closer and closer to that of $R_\rmL[f_{ \rmI},g_T]$ for both examples. To quantify the convergence, we define
\begin{equation*}
\mathrm{Err}_\rmR(\eps) = \frac{\|R^\eps[f_{ \rmI}^\eps,g_T^\eps]-R[f_{ \rmI},g_T]\|_{\LdC}}{\|R[f_{ \rmI},g_T]\|_{\LdC}}\,,
\end{equation*}
and plot the convergence rate with respect to $\eps$, as shown in the right column of Figure~\ref{fig:R_err}. In both examples, the plots suggest a decay rate of $O(\varepsilon^2)$.

\begin{figure}[tbhp]
  \centering
  \subfloat[$b_x = c_x = 0.1875$]{\label{fig:R_f1_g1}
  \includegraphics[width=0.45\textwidth,height = 0.18\textheight]{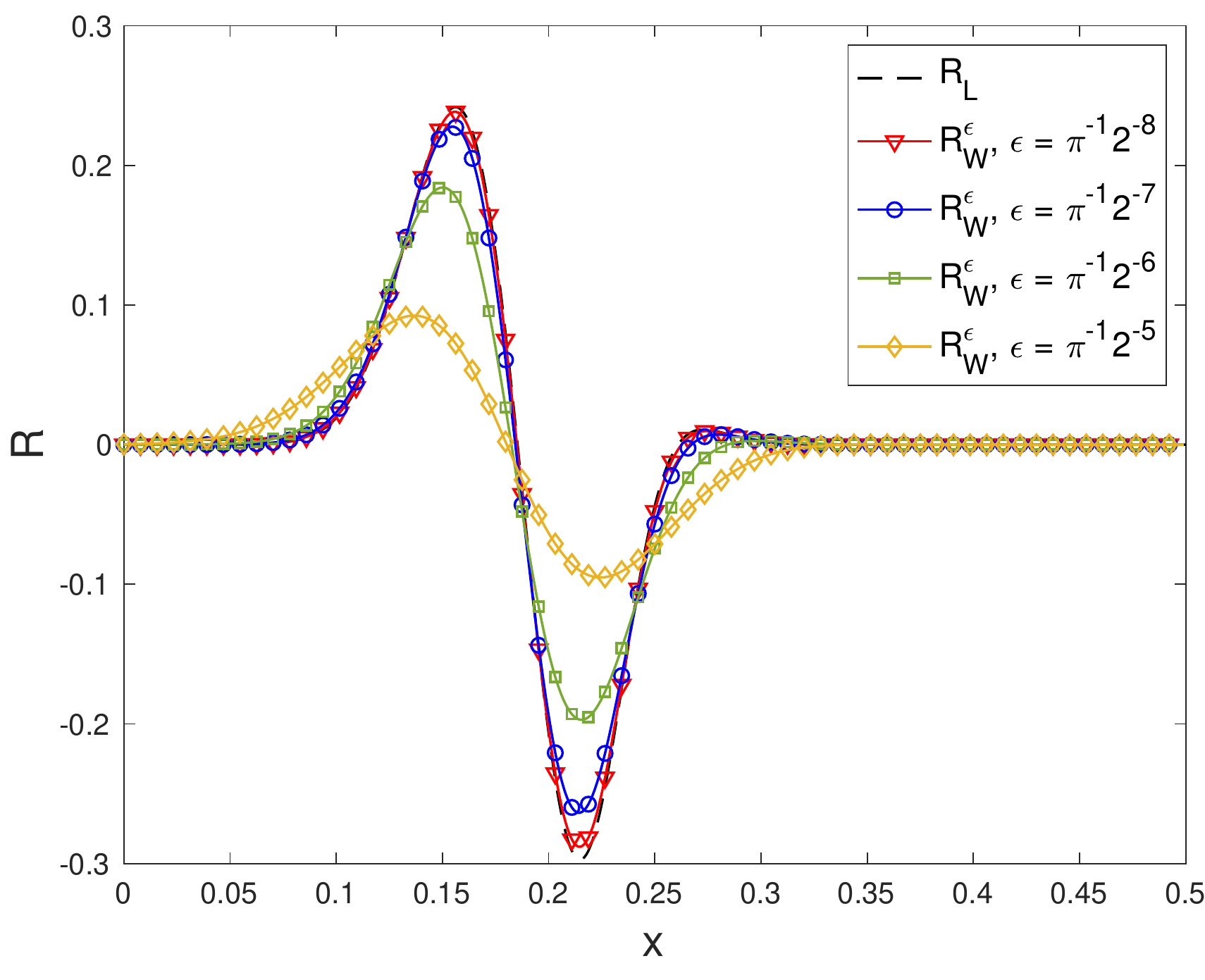}
  \includegraphics[width=0.45\textwidth,height = 0.18\textheight]{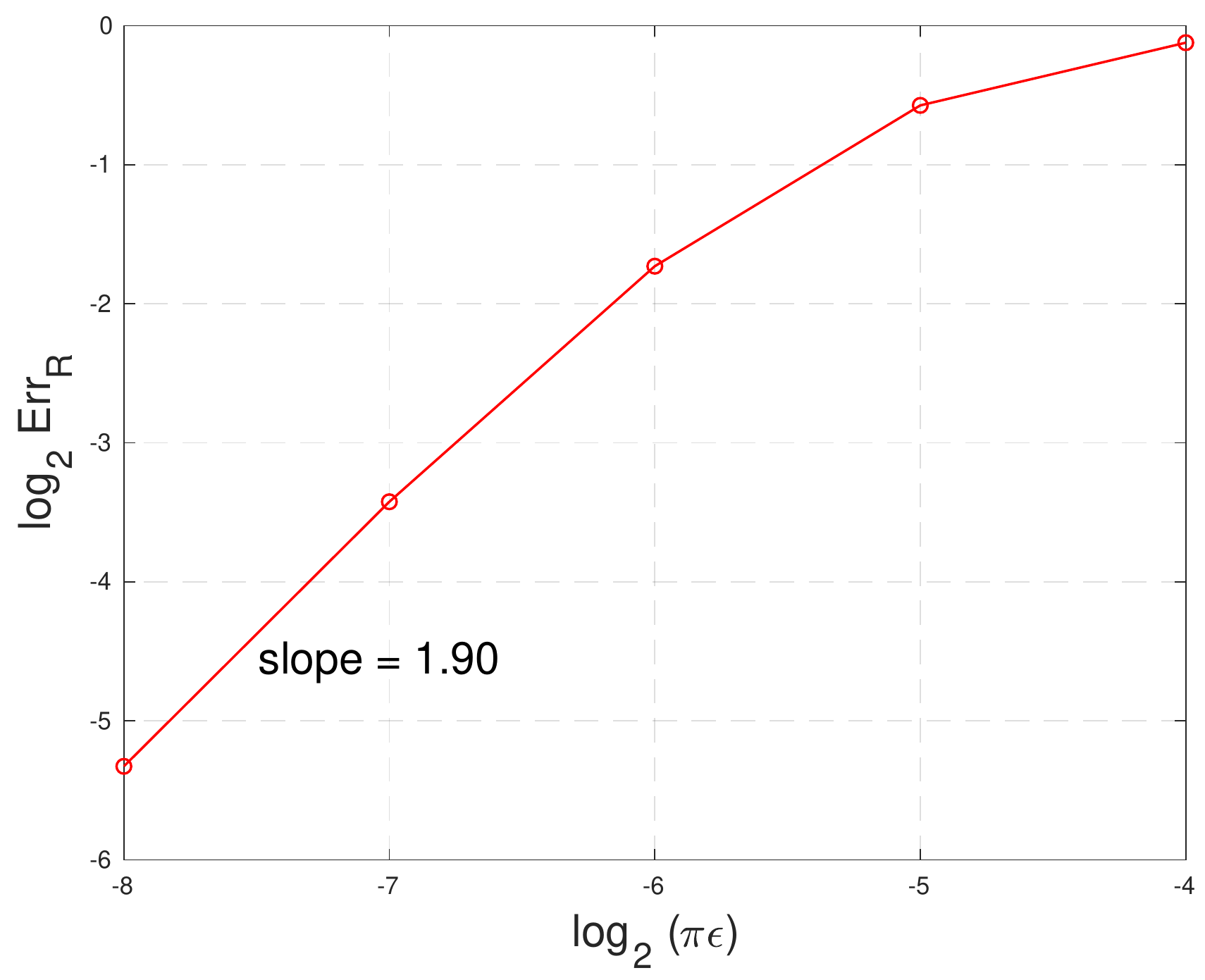}
  }

  \subfloat[$b_x = c_x = 0.25$]{\label{fig:R_fm_gm}
  \includegraphics[width=0.45\textwidth,height = 0.18\textheight]{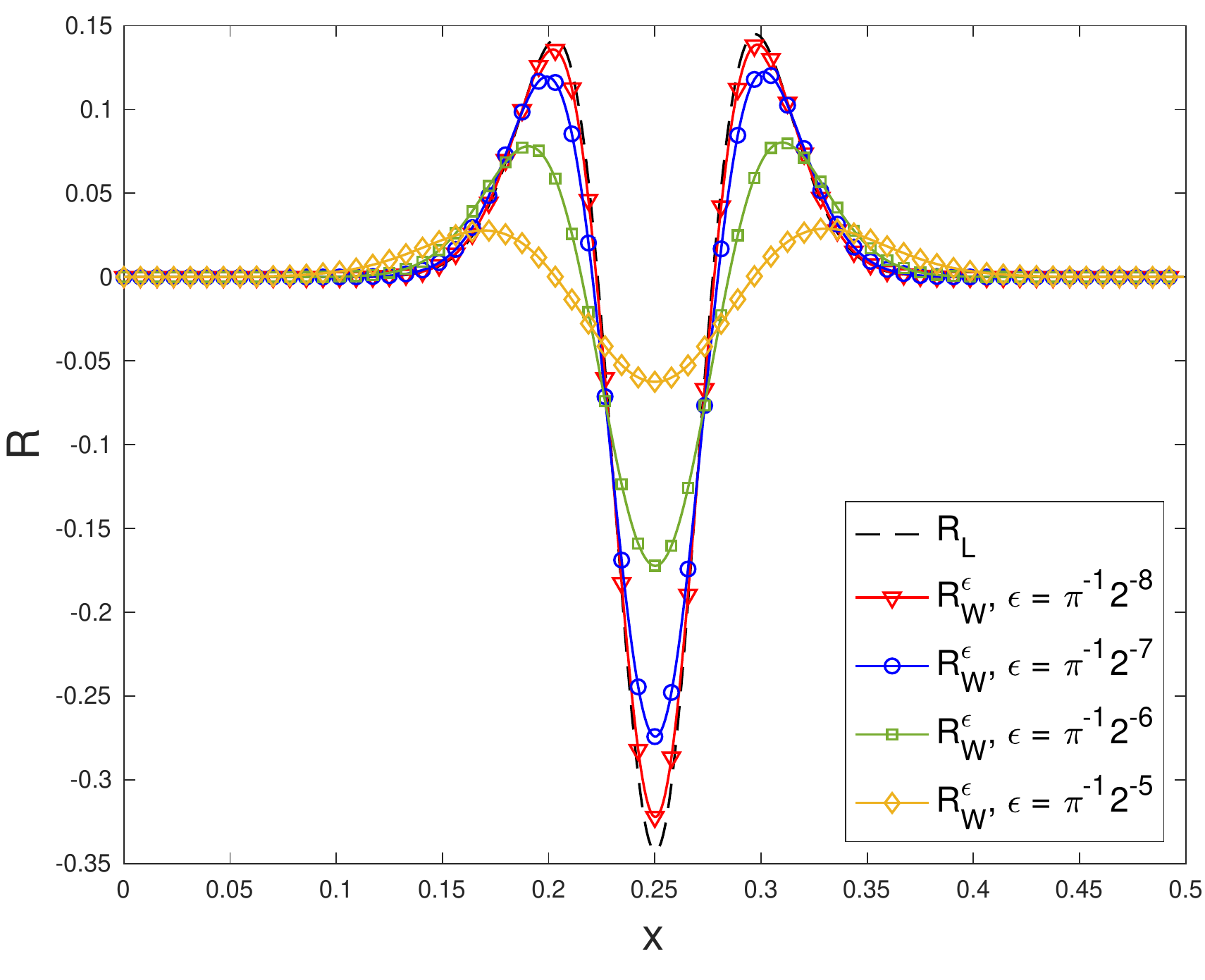}
  \includegraphics[width=0.45\textwidth,height = 0.18\textheight]{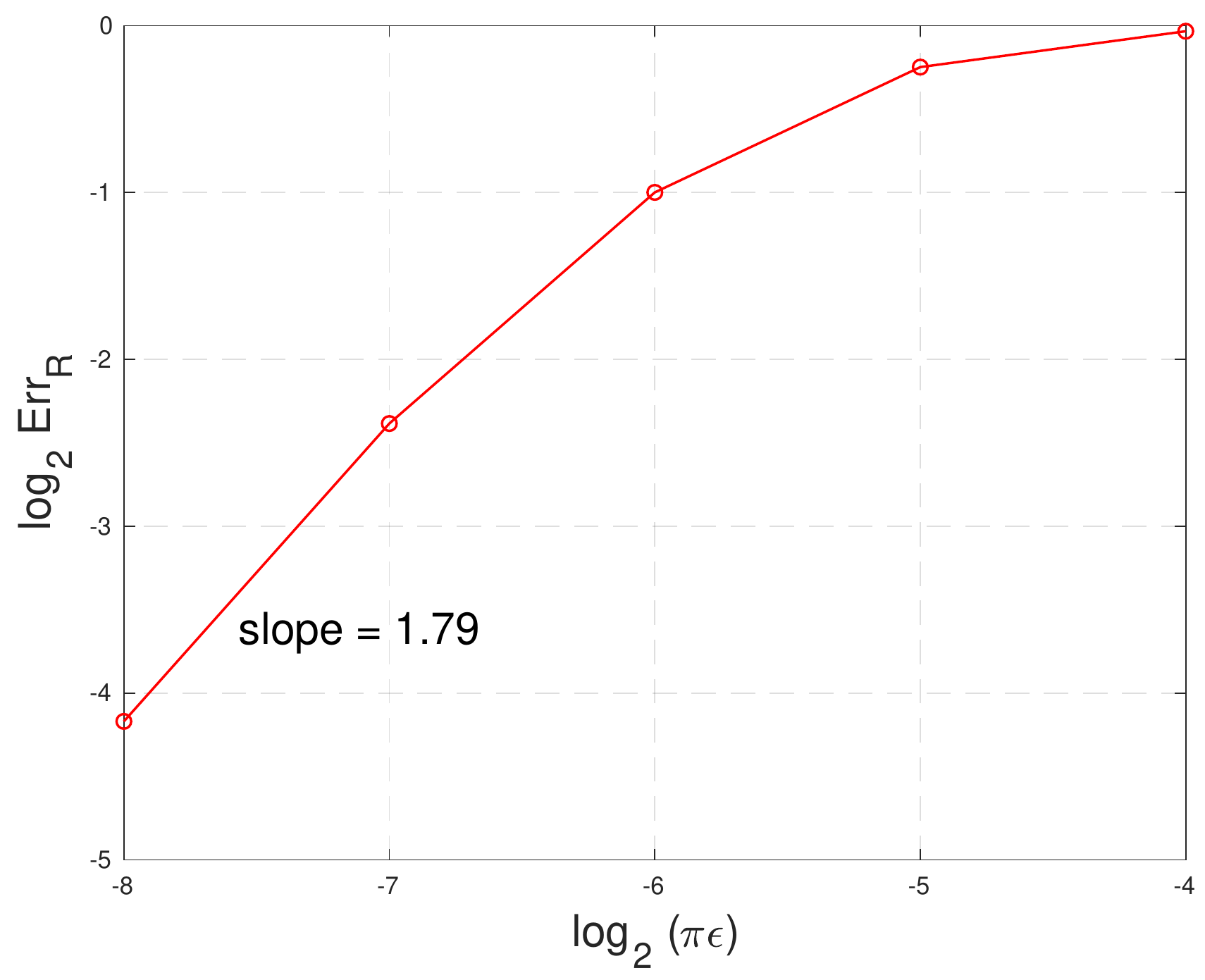}
  }

  \caption{The left column compares the Wigner representative $R_\rmW^\eps[f_{\rmb \rmI}^\eps,g_T^\eps]$ with different values of $\eps$ and the limiting Liouville representative $R_\rmL[f_{\rmb \rmI},g_T]$. The right column shows $\mathrm{Err}_\rmR(\eps)$ as a function of $\eps$. The decay rate suggests that $\mathrm{Err}_\rmR(\eps)$ is of $O(\varepsilon^2)$.} \label{fig:R_err}
\end{figure}

We finally demonstrate the convergence for a large set of basis functions. To do so, we first take the interval $[x_l,x_r] = [0.1875,0.3125]$, and denote the discrete points in the interval $x_i = x_l + i\Delta x$, with $i\in\Nb$ and $x_l\leq x_i\leq x_r$. Considering $\Delta x = 2^{-10}$, we have $N = 129$ configurations of $x_i$. We then let $b_{x}$ and $c_x$, the centers for $f^\eps_{ \rmI}$ and $g^\eps_T$ taking these configurations. The combination provides us a large set of initial/final time data $f_{ \rmI,i}^\eps$ and $g_{T,j}^\eps$. We compute the corresponding solutions, termed $f^\eps_i$ and $g^\eps_{j}$ and formulate a set:
\begin{equation*}
\mathcal{R}^\eps_\rmW=\{R^\eps_{\rmW,ij} = R^\eps_\rmW[f_{ \rmI,i}^\eps\,,g^\eps_{T,j}]\,,\;i,j=0\,,\cdots\,,128 \}\,.
\end{equation*}
The same process is done to obtain $R_{\rmL,ij}$ and the set $\mathcal{R}_\rmL$.

We now compare the set $\mathcal{R}_\rmW^\eps$ and $\mathcal{R}_\rmL$. We first compare the singular values of the two sets. In Figure~\ref{fig:svd}, we plot the relative singular value decay of both $\mathcal{R}_\rmW^\eps$ at different values of $\eps$, and $\mathcal{R}_\rmL$. As $\eps\to0$, it is clear the decay profile converges. We also quantify the convergence of relative singular value using the following error term:
\begin{equation*}
\mathrm{Err}_{s,i}(\eps) = \frac{|s_i^\eps-s_i|}{|s_i|}\,,
\end{equation*}
where $s_i^\eps$ is the $i$th relative singular value of $\mathcal{R}^\eps_\rmW$, and $s_i$ is the $i$th relative singular value of $\mathcal{R}_\rmL$. In Figure~\ref{fig:err_s}, we plot $\mathrm{Err}_{s_i}(\eps)$ as a function of $\eps$ for $i = 2, 3, 4, 5$. It is clear that the relative singular values of $\mathcal{R}^\eps_\rmW$ converge to their counterparts in the $\eps\to0$ classical limit.

\begin{figure}[tbhp]
  \centering
  \subfloat[]{\label{fig:svd}
  \includegraphics[width=0.45\textwidth,height = 0.18\textheight]{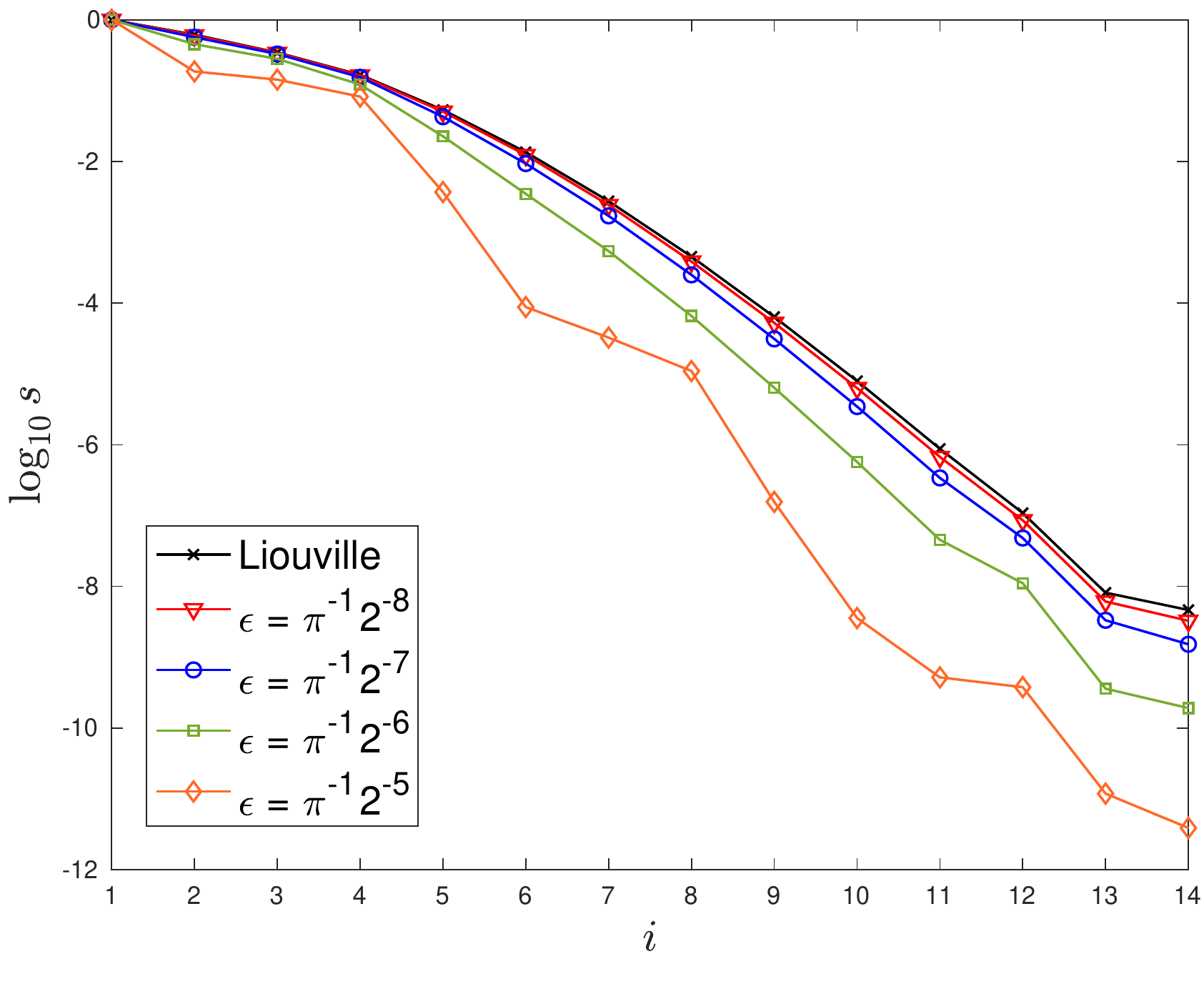}
  }
  \subfloat[]{\label{fig:err_s}
  \includegraphics[width=0.45\textwidth,height = 0.18\textheight]{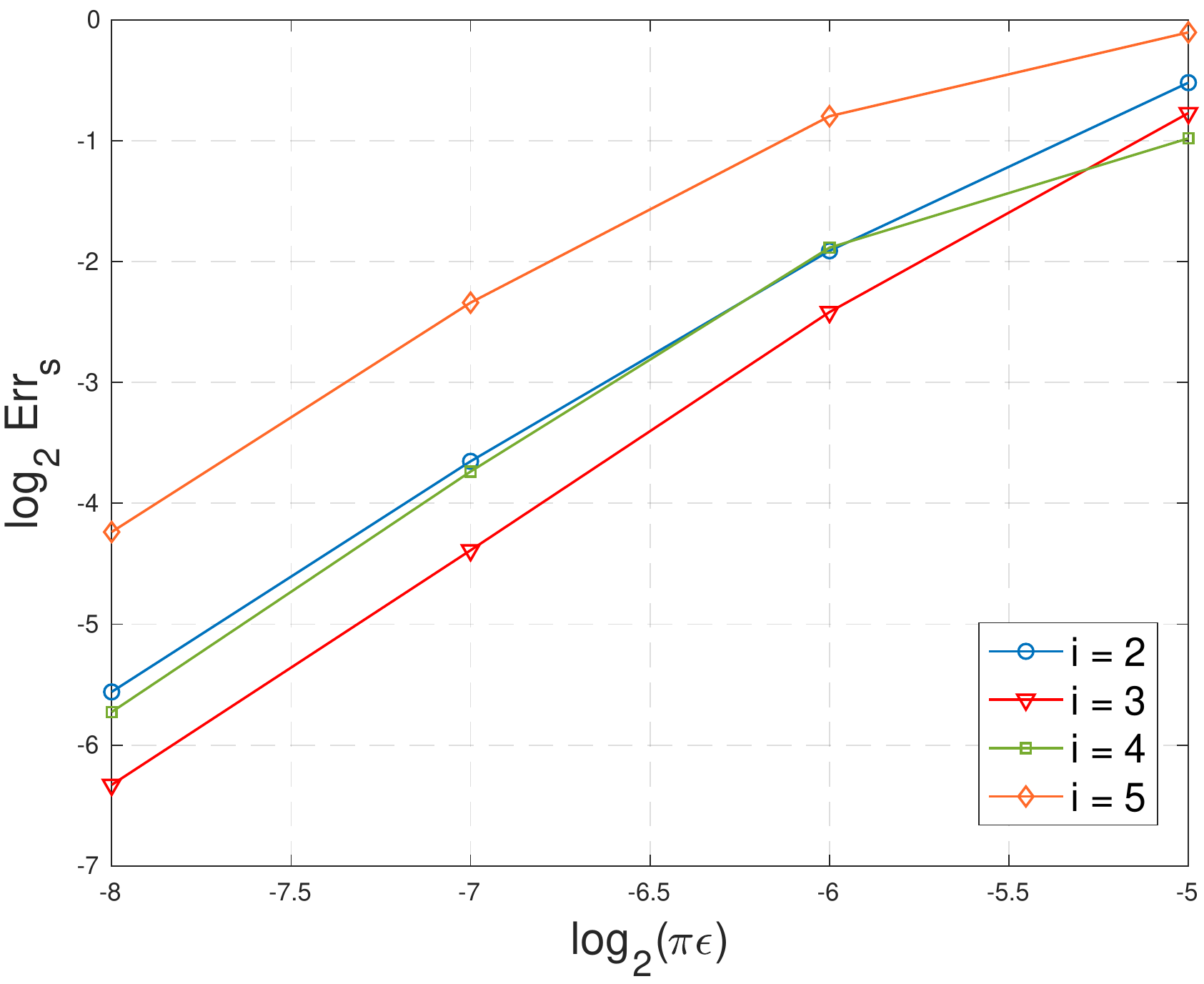}
  }
  \caption{(a) The relative singular values of $\mathcal{R}_\rmL$ and $\mathcal{R}^\eps_\rmW$ at different values of $\eps$. (b) $\mathrm{Err}_{s_i}(\eps)$ as a function of $\eps$ for the 2nd to the 5th relative singular values.}
\end{figure}

We then compare the left singular vectors of the basis. In Figure~\ref{fig:eig_vec}, we show the first, third, seventh and tenth left singular vectors of $\mathcal{R}^\eps_\rmW$. As $\eps\to0$, the profiles converge to those of $\mathcal{R}_\rmL$. To quantify such convergence, we let $Q_k^\eps$ and $Q_k$ to denote the column spaces (orthonormalized) spanned by the first $k$ left singular vectors of $\mathcal{R}^\eps_\rmW$ and $\mathcal{R}_\rmL$, respectively, and define the angle between the spaces:
\begin{equation*}
\mathrm{Err}_{\mathcal{R},k}= \|Q_k - Q_k^\eps (Q_k^\eps)^\top Q_k\|_2\,.
\end{equation*}
The angle between the two spaces are shown to converge as $\eps\to0$ for different values of $k$ in Figure~\ref{fig:angle}.

\begin{figure}[tbhp]
  \centering
  \subfloat[The first singular vector]{
  \includegraphics[width=0.45\textwidth,height = 0.18\textheight]{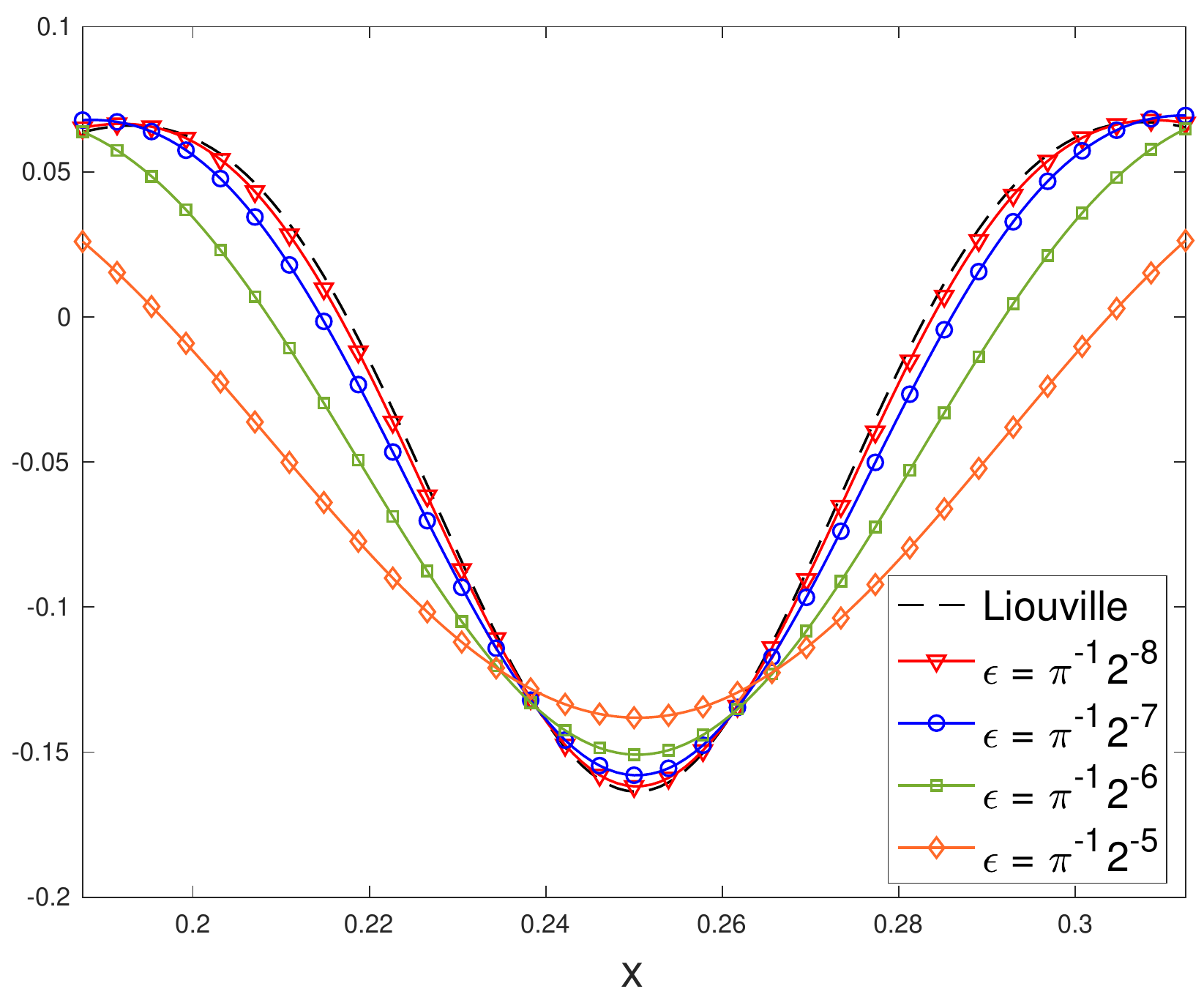}
  }
  \subfloat[The third singular vector]{
  \includegraphics[width=0.45\textwidth,height = 0.18\textheight]{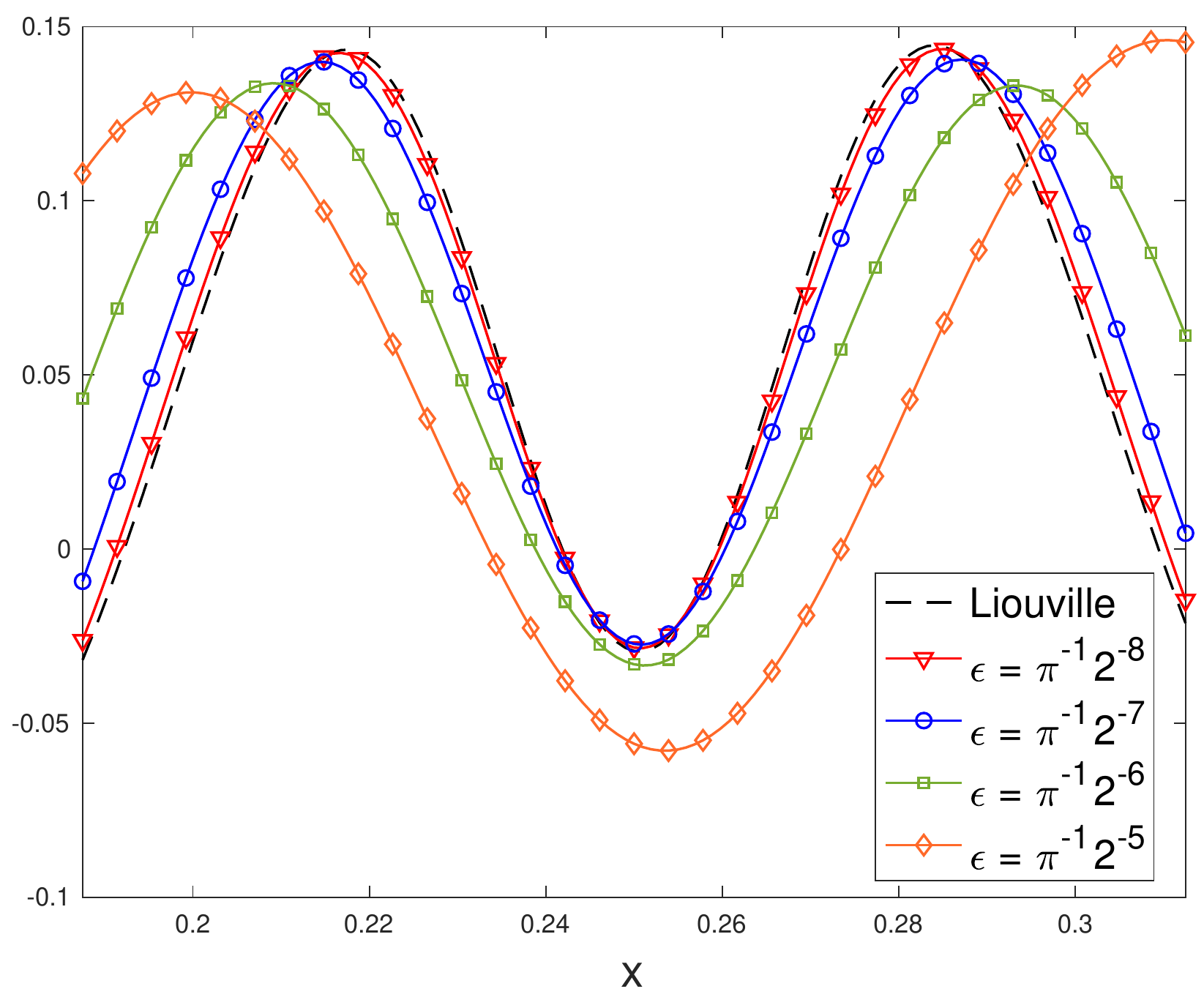}
  }

  \subfloat[The seventh singular vector]{
  \includegraphics[width=0.45\textwidth,height = 0.18\textheight]{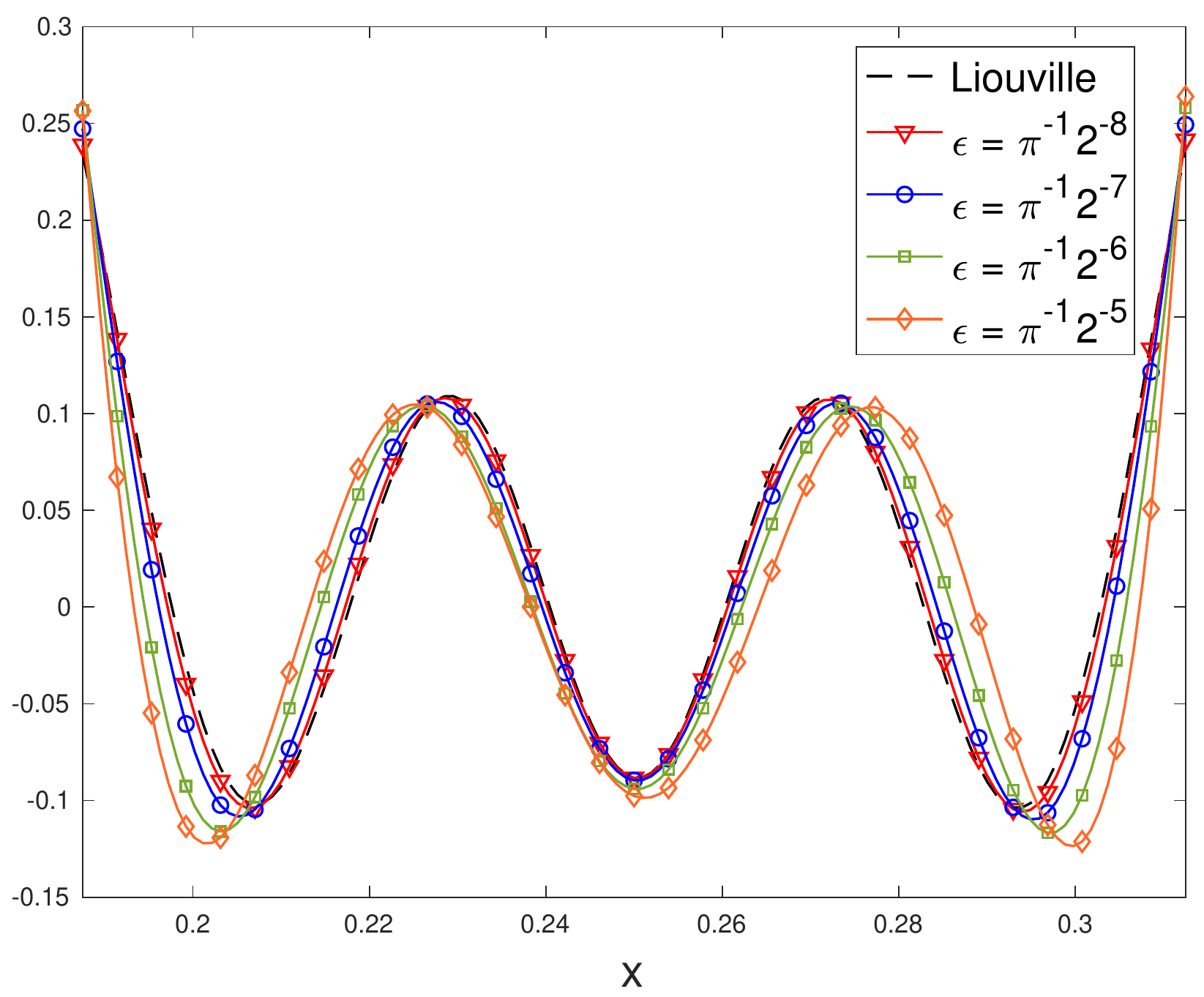}
  }
  \subfloat[The tenth singular vector]{
  \includegraphics[width=0.45\textwidth,height = 0.18\textheight]{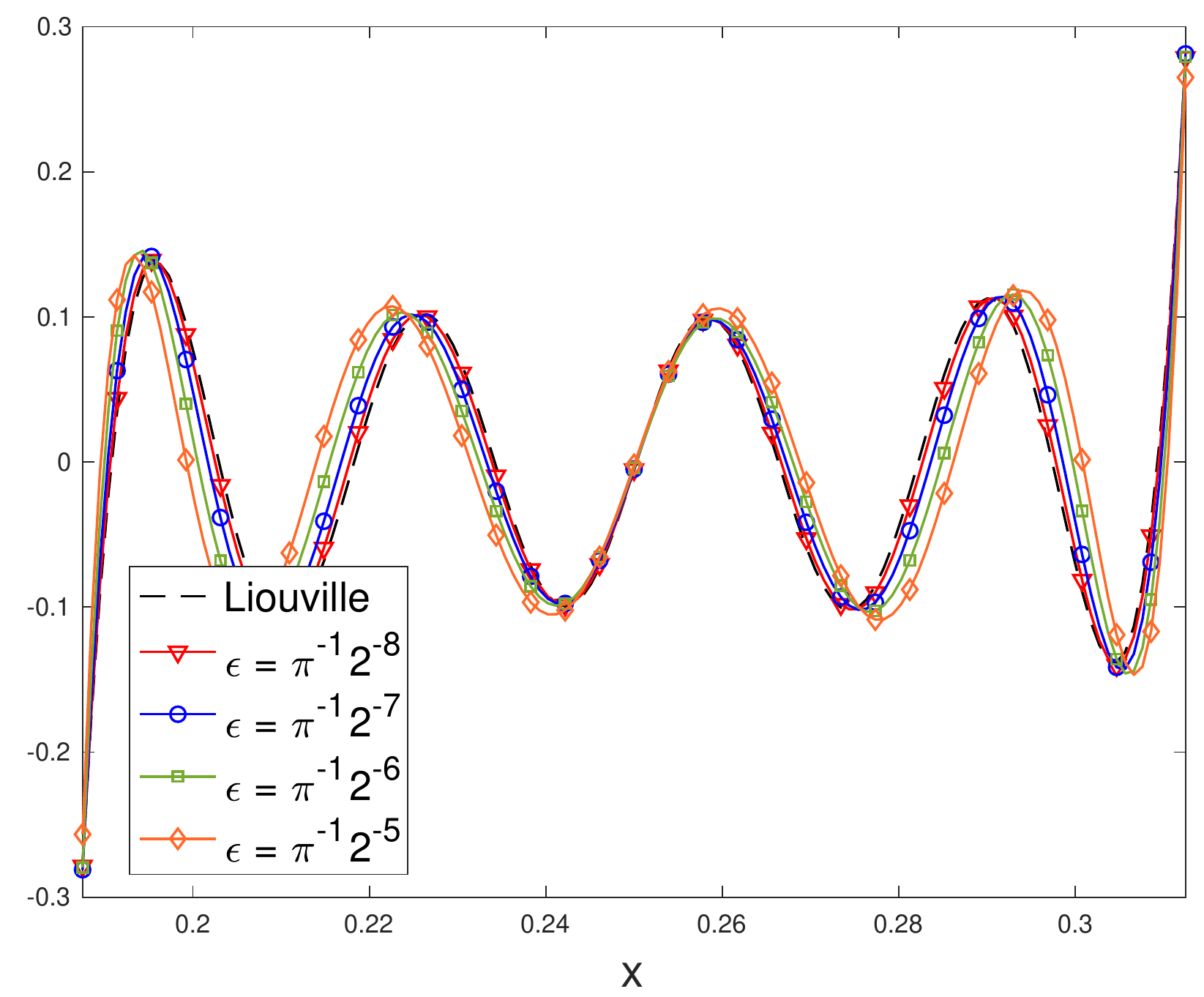}
  }

  \caption{The singular vectors of $\mathcal{R}_\rmL$ and $\mathcal{R}^\eps_\rmW$ at different values of $\eps$.}\label{fig:eig_vec}
\end{figure}

\begin{figure}[tbhp]
  \centering
  \includegraphics[width=0.45\textwidth,height = 0.18\textheight]{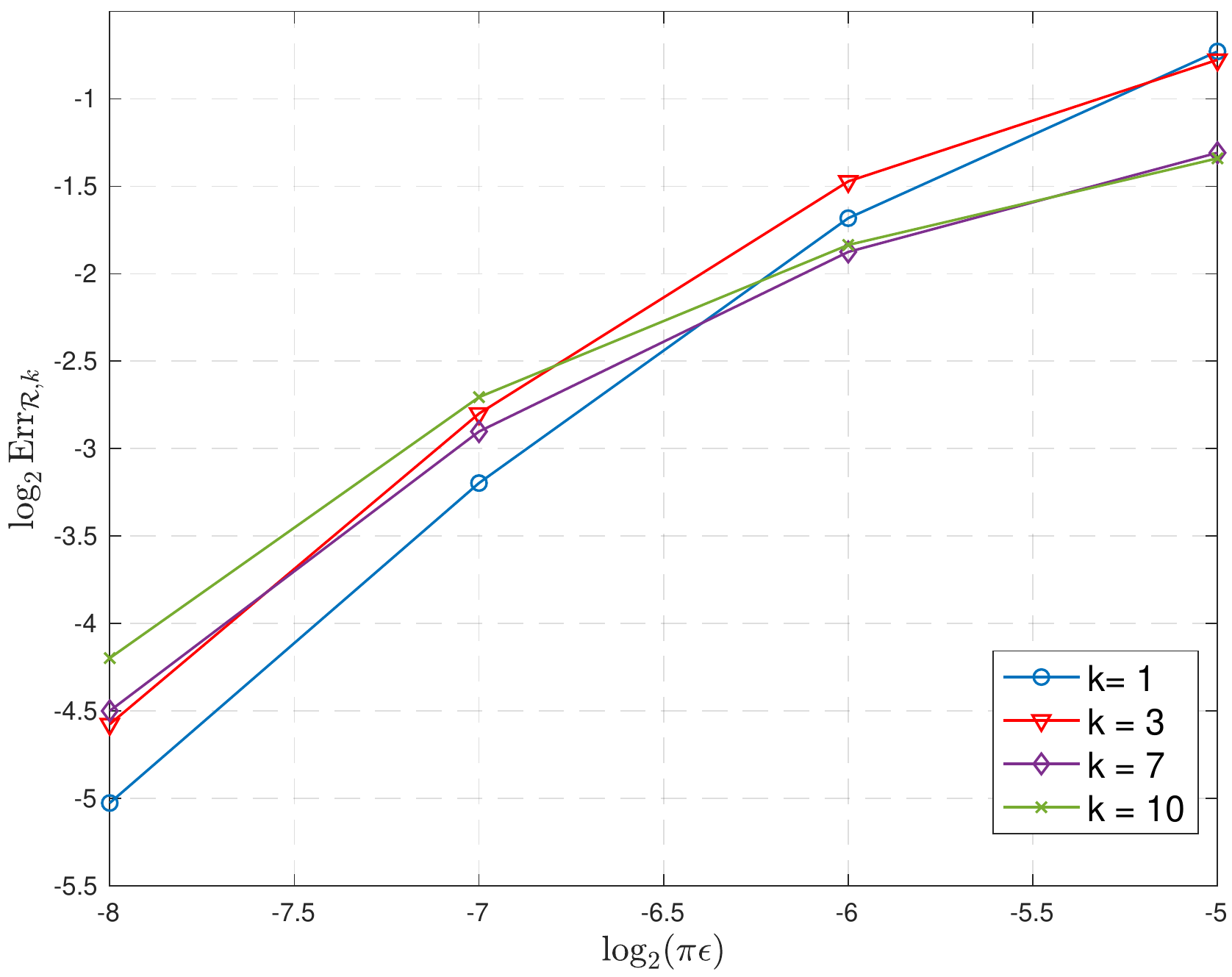}

  \caption{$\mathrm{Err}_{\mathcal{R},k}$ as a function of $\eps$ for $k=1, 3, 7, 10$.}\label{fig:angle}
\end{figure}

\section{Conclusion}\label{sec:conclusion}
It is a well-known result that the Schr\"odinger equation leads to the Newton's second law in the classical limit, when the rescaled Planck constant $\eps\to0$. We investigate this limit in the inverse setting. More specifically, we assume the initial and final data is available and we study if the initial-final data pairs can reconstruct the potential term in the Schr\"odinger equation. The investigation is done in the linearized setting, assuming the potential is close enough to a background, and this boils the problem down to the study of the representative of the associated Fredholm integral derived from the inverse problem.

We employ the Wigner transform tool. In particular, we translate the information of the Schr\"odinger equation to that of the Wigner equation, and pass its limit to obtain the Liouville equation which presents particle trajectories following the classical mechanics. We are able to show that the representative derived under the Wigner framework indeed converges to the representative derived under the Liouville framework when $\eps\to0$, and thus we link the inverse Schr\"odinger problem with the inverse Newton's law of motion.

%\begin{acknowledgements}
%\end{acknowledgements}

\section*{Conflict of interest}
The authors declare that they have no conflict of interest.


\bibliographystyle{siamplain}
\bibliography{InverseSchr}

\begin{thebibliography}{10}

\bibitem{Ba:2013}
{\sc G.~Bal}, {\em Hybrid inverse problems and internal functionals}, Inverse
  problems and applications: inside out. II, 60 (2013), pp.~325--368.

\bibitem{BaKoRy:2010}
{\sc G.~Bal, T.~Komorowski, and L.~Ryzhik}, {\em Kinetic limits for waves in a
  random medium}, Kinet. Relat. Models, 3 (2010), pp.~529--644.

\bibitem{BaPi:2007}
{\sc G.~Bal and O.~Pinaud}, {\em Kinetic models for imaging in random media},
  Multiscale Model. Simul., 6 (2007), pp.~792--819.

\bibitem{BaPiRy:2015}
{\sc G.~Bal, O.~Pinaud, and L.~Ryzhik}, {\em Random Media in Inverse Problems,
  Theoretical Aspects}, Springer Berlin Heidelberg, Berlin, Heidelberg, 2015,
  pp.~1219--1222.

\bibitem{BaRe:2008}
{\sc G.~Bal and K.~Ren}, {\em Transport-based imaging in random media}, SIAM J.
  Appl. Math., 68 (2008), pp.~1738--1762.

\bibitem{BaRe:2011}
{\sc G.~Bal and K.~Ren}, {\em Multi-source quantitative photoacoustic
  tomography in a diffusive regime}, Inverse Problems, 27 (2011), p.~075003.

\bibitem{BaReUhZh:2011}
{\sc G.~Bal, K.~Ren, G.~Uhlmann, and T.~Zhou}, {\em Quantitative
  thermo-acoustics and related problems}, Inverse Problems, 27 (2011),
  p.~055007.

\bibitem{Be:2011}
{\sc P.~Beard}, {\em Biomedical photoacoustic imaging}, Interface focus, 1
  (2011), pp.~602--631.

\bibitem{Ca:2006}
{\sc A.~P. Calder{\'o}n}, {\em On an inverse boundary value problem}, Comput.
  Appl. Math., 25 (2006), pp.~133--138.

\bibitem{CaGaMaSh:2003}
{\sc J.~Carrillo, I.~Gamba, A.~Majorana, and C.-W. Shu}, {\em A {WENO}-solver
  for the transients of {B}oltzmann--{P}oisson system for semiconductor
  devices: performance and comparisons with {M}onte {C}arlo methods}, J.
  Comput. Phys., 184 (2003), pp.~498--525.

\bibitem{CaSc:2015}
{\sc A.~Caz{\'e} and J.~Schotland}, {\em Diagrammatic and asymptotic approaches
  to the origins of radiative transport theory: tutorial}, JOSA A, 32 (2015),
  pp.~1475--1484.

\bibitem{ChLiWa:2018}
{\sc K.~Chen, Q.~Li, and L.~Wang}, {\em Stability of stationary inverse
  transport equation in diffusion scaling}, Inverse Problems, 34 (2018),
  p.~025004.

\bibitem{GeMaMaPo:1997}
{\sc P.~G{\'e}rard, P.~Markowich, N.~Mauser, and F.~Poupaud}, {\em
  Homogenization limits and {W}igner transforms}, Comm. Pure Appl. Math., 50
  (1997), pp.~323--379.

\bibitem{Go:2005}
{\sc J.~Goodman}, {\em Introduction to {F}ourier optics}, Roberts and Company
  Publishers, 2005.

\bibitem{HoMoSt:2018}
{\sc S.~Holman, F.~Monard, and P.~Stefanov}, {\em The attenuated geodesic
  {X}-ray transform}, Inverse Problems, 34 (2018), p.~064003.

\bibitem{Ho:85}
{\sc L.~H{\"o}rmander}, {\em The Analysis of Linear Partial Differential
  Operators {III}}, Springer Berlin Heidelberg, 1985.

\bibitem{HoKrSc:2018}
{\sc J.~Hoskins, J.~Kraisler, and J.~Schotland}, {\em Radiative transport in
  quasi-homogeneous random media}, JOSA A, 35 (2018), pp.~1855--1860.

\bibitem{JiSh:1996efficient}
{\sc G.-S. Jiang and C.-W. Shu}, {\em Efficient implementation of weighted
  {ENO} schemes}, J. Comput. Phys., 126 (1996), pp.~202--228.

\bibitem{Jo:2013}
{\sc A.~Jollivet}, {\em On inverse scattering at fixed energy for the
  multidimensional {N}ewton equation in a non-compactly supported field}, J.
  Inverse Ill-Posed Probl., 21 (2013), pp.~713--734.

\bibitem{Jo:2014}
{\sc A.~Jollivet}, {\em Inverse scattering at high energies for the
  multidimensional {N}ewton equation in a long range potential}, Asymptot.
  Anal., 90 (2014), pp.~105--132.

\bibitem{KeKaSh:1956}
{\sc J.~Keller, I.~Kay, and J.~Shmoys}, {\em Determination of the potential
  from scattering data}, Phys. Rev., 102 (1956), pp.~557--559.

\bibitem{LaLiUh:2019}
{\sc R.-Y. Lai, Q.~Li, and G.~Uhlmann}, {\em Inverse problems for the
  stationary transport equation in the diffusion scaling}, SIAM J. Appl. Math.,
  79 (2019), pp.~2340--2358.

\bibitem{Mo:2014}
{\sc F.~Monard}, {\em Numerical implementation of geodesic {X}-ray transforms
  and their inversion}, SIAM J. Imaging Sci., 7 (2014), pp.~1335--1357.

\bibitem{MoStUh:2015}
{\sc F.~Monard, P.~Stefanov, and G.~Uhlmann}, {\em The geodesic ray transform
  on {R}iemannian surfaces with conjugate points}, Comm. Math. Phys., 337
  (2015), pp.~1491--1513.

\bibitem{Mu:1981}
{\sc R.~G. Mukhometov}, {\em A problem of reconstructing a {R}iemannian
  metric}, Sib. Math. J., 22 (1981).

\bibitem{NaUhWa:2013}
{\sc S.~Nagayasu, G.~Uhlmann, and J.-N. Wang}, {\em Increasing stability in an
  inverse problem for the acoustic equation}, Inverse Problems, 29 (2013),
  p.~025012.

\bibitem{No:1999}
{\sc R.~Novikov}, {\em Small angle scattering and {X}-ray transform in
  classical mechanics}, Ark. Mat., 37 (1999), pp.~141--169.

\bibitem{RyPaKe:1996}
{\sc L.~Ryzhik, G.~Papanicolaou, and J.~Keller}, {\em Transport equations for
  elastic and other waves in random media}, Wave motion, 24 (1996),
  pp.~327--370.

\bibitem{Sz:2004}
{\sc T.~Szabo}, {\em Diagnostic ultrasound imaging: inside out}, Academic
  Press, 2004.

\bibitem{Wa:1999}
{\sc J.-N. Wang}, {\em Stability estimates of an inverse problem for the
  stationary transport equation}, in Ann. Inst. H. Poincar\'e Phys. Th\'eor.,
  vol.~70, 1999, pp.~473--495.

\end{thebibliography}

\end{document}